\documentclass[12pt,a4paper]{article}
\usepackage[paperheight=297mm,paperwidth=210mm,top=20mm,bottom=20mm,right=20mm,left=20mm,heightrounded]{geometry}

\usepackage[utf8]{inputenc}
\usepackage[english]{babel}

\usepackage{cmap}  \usepackage[T2A]{fontenc}
\usepackage[pdftex,unicode=true,bookmarks=true]{hyperref}

\usepackage{xr}
\usepackage{indentfirst}

\usepackage{graphicx}
\usepackage[margin=10pt,font=small,labelfont=bf,labelsep=endash]{caption}

\usepackage{floatflt}

\usepackage{latexsym}
\usepackage{amsfonts}
\usepackage{amsmath}
\usepackage{amssymb}

\usepackage{cancel} 

\usepackage{float}

\usepackage{amsthm}
\usepackage{wrapfig}

\newtheorem*{definition}{Definition}
\newtheorem{lemma}{Lemma}

\makeatletter
\def\th@plain{%
	\thm@notefont{}
	\itshape 
}
\def\th@definition{%
	\thm@notefont{}
	\normalfont 
}
\makeatother

\newcommand{\I}{\mathrm{i}}
\newcommand{\E}{\mathrm{e}}
\newcommand{\md}[1]{~~(\mathrm{mod}~#1)}
\newcommand{\be}{\begin{equation}}
\newcommand{\ee}{\end{equation}}
\newcommand{\bean}{\begin{eqnarray*}}
\newcommand{\eean}{\end{eqnarray*}}

\newcommand{\bea}{\begin{eqnarray}}
\newcommand{\eea}{\end{eqnarray}}

\newcommand{\zrf}{\operatorname{zrf}}
\newcommand{\lattice}{\operatorname{lattice}}
\newcommand{\ncells}{\operatorname{ncells}}

\DeclareSymbolFont{greek}{U}{eur}{m}{n}
\SetSymbolFont{greek}{bold}{U}{eur}{b}{n}

\DeclareMathSymbol{\GGamma}{\mathord}{greek}{"00}
\DeclareMathSymbol{\DDelta}{\mathord}{greek}{"01}
\DeclareMathSymbol{\TTheta}{\mathord}{greek}{"02}
\DeclareMathSymbol{\LLambda}{\mathord}{greek}{"03}
\DeclareMathSymbol{\XXi}{\mathord}{greek}{"04}
\DeclareMathSymbol{\PPi}{\mathord}{greek}{"05}
\DeclareMathSymbol{\SSigma}{\mathord}{greek}{"06}
\DeclareMathSymbol{\UUpsilon}{\mathord}{greek}{"07}
\DeclareMathSymbol{\PPhi}{\mathord}{greek}{"08}
\DeclareMathSymbol{\OOmega}{\mathord}{greek}{"09}
\DeclareMathSymbol{\OOmega}{\mathord}{greek}{"0A}

\DeclareMathSymbol{\aalpha}{\mathord}{greek}{"0B}
\DeclareMathSymbol{\bbeta}{\mathord}{greek}{"0C}
\DeclareMathSymbol{\ggamma}{\mathord}{greek}{"0D}
\DeclareMathSymbol{\ddelta}{\mathord}{greek}{"0E}
\DeclareMathSymbol{\eepsilon}{\mathord}{greek}{"0F}
\DeclareMathSymbol{\zzeta}{\mathord}{greek}{"10}
\DeclareMathSymbol{\eeta}{\mathord}{greek}{"11}
\DeclareMathSymbol{\ttheta}{\mathord}{greek}{"12}
\DeclareMathSymbol{\iiota}{\mathord}{greek}{"13}
\DeclareMathSymbol{\kkappa}{\mathord}{greek}{"14}
\DeclareMathSymbol{\llambda}{\mathord}{greek}{"15}
\DeclareMathSymbol{\mmu}{\mathord}{greek}{"16}
\DeclareMathSymbol{\nnu}{\mathord}{greek}{"17}
\DeclareMathSymbol{\xxi}{\mathord}{greek}{"18}
\DeclareMathSymbol{\ppi}{\mathord}{greek}{"19}
\DeclareMathSymbol{\rrho}{\mathord}{greek}{"1A}
\DeclareMathSymbol{\ssigma}{\mathord}{greek}{"1B}
\DeclareMathSymbol{\ttau}{\mathord}{greek}{"1C}
\DeclareMathSymbol{\uupsilon}{\mathord}{greek}{"1D}
\DeclareMathSymbol{\pphi}{\mathord}{greek}{"1E}
\DeclareMathSymbol{\cchi}{\mathord}{greek}{"1F}
\DeclareMathSymbol{\ppsi}{\mathord}{greek}{"20}
\DeclareMathSymbol{\oomega}{\mathord}{greek}{"21}
\DeclareMathSymbol{\vvarepsilon}{\mathord}{greek}{"22}
\DeclareMathSymbol{\vvartheta}{\mathord}{greek}{"23}
\DeclareMathSymbol{\vvarpi}{\mathord}{greek}{"24}

\title{Number-theory renormalization of vacuum energy}
\author{M.G. Ivanov\thanks{\tt \href{mailto:ivanov.mg@mipt.ru}{ivanov.mg@mipt.ru}}, 
	V.A. Dudchenko\thanks{\tt \href{mailto:dudchenko.va@phystech.edu}{dudchenko.va@phystech.edu}}, 
	V.V. Naumov\thanks{\tt \href{mailto:naumov.vv@phystech.edu}{\textbf{naumov.vv@phystech.edu}}}
	\\ ${}^{*\ddagger}$ Moscow Institute of Physics and Technology
	\\ ${}^\dagger$ Vernadsky Institute of Geochemistry and Analytical Chemistry\\ of the Russian Academy of Sciences}

\date{July 21, 2023}

\begin{document}
\maketitle

\begin{abstract}
For QFT on a lattice of dimension $d\geqslant 3$, the vacuum energy (both bosonic and fermionic) is zero if
the Hamiltonian is a function of the square of the momentum, and the calculation of the vacuum energy is performed in the ring  of residue classes modulo $N$. 
This fact is related to a problem from number theory about the number of ways to represent a number as a sum of $d$ squares in the ring  of residue classes modulo $N$.
\end{abstract}

\tableofcontents

\newpage

\section{Introduction}

In most models of quantum field theory (QFT), the problem of renormalization of the vacuum energy \cite{vac-en} arises.
We assume that the problem may be related to the uncritical use of real numbers and operations on them in QFT. 
The possibility of using other numerical systems in QFT was raised in the monograph \cite{p-adic-vvz}, in which $p$-adic numbers from the methods of number theory were transferred to mathematical physics.
The role of number theory in physics was discussed in the paper \cite{volovich}.

We are considering QFT on the lattice. 
But we do not consider the transition from a continuous space to a lattice using difference schemes.
Instead, we build a theory on a lattice using arithmetic operations in the ring $\mathbb{Z}(N)$ of residue classes modulo $N$, assuming that such arithmetic is native for this lattice.
In this approach the renormalization of the vacuum energy occurs naturally for a wide class of models under consideration.

Initially, the problem was set (M.I.) to study the vacuum energy for the dispersion relation $E(\mathbf{p})=\sqrt{\mathbf{p}^2+m^2}$ with different definitions of square root in the ring $\mathbb{Z}(N)$.
Numerical calculation (V.N.) showed that, regardless of the method of determining the lattice analogue of the positive branch of the square root, the vacuum energy is zero at any $N$ with the dimension of the space $d\geqslant3$.
After that, the multiplicities of various values of $\mathbf{p}^2$ in the ring $\mathbb{Z}(N)$ were numerically calculated depending on $N$ and the dimension of the space $d$ (V.N), it turned out that all multiplicities are divisible by $N$, for arbitrary $N$ if $d\geqslant3$, and also if $N=2^n$ and $d\geqslant2$. This means that all multiplicities of different values of $\mathbf{p}^2$ are zero modulo $N$, and all values of $\mathbf{p}^2\in\mathbb{Z}(N)$ contribute zero to the vacuum energy.
The corresponding theorem is formulated and proved in this paper.

\section{Theorem on number-theory renormalization}

We consider a bosonic quantum field theory on a lattice $\mathbb{Z}^d(N)$ with a Hamiltonian of the form
\bea\label{H_b}
&&\hat H_{b}=\sum_{\mathbf{p}\in\mathbb{Z}^d(N),\mathbf{p}^2\in D}E(\mathbf{p}^2)\,\left(2\hat b_{\mathbf{p}}^\dagger\hat b_{\mathbf{p}}+1\right),\\
\nonumber
&&
\mathbf{p}^2=\sum_{k=1}^d p_k^2\in\mathbb{Z}(N),\quad
E:D\to\mathbb{Z}(N),\quad 
D\subset\mathbb{Z}(N),\quad
[\hat b_{\mathbf{p_1}},\hat b_{\mathbf{p_2}}^\dagger]=\ddelta_{\mathbf{p_1}\mathbf{p_2}}\,\hat 1.
\eea
We believe permissible $\mathbf{p}\in\mathbb{Z}^d(N)$ iff $\mathbf{p}^2\in D$.

Hereafter $\mathbb{Z}(N)$ is the ring  of residue classes modulo $N$. 
Usually we will use a representation of the form $\mathbb{Z}(N)=\{0,1,2,\dots,N-1\}$.

We will need the operation of \textit{division with remainder}:
\be\label{def-mod}
  a\mod N \text{~~is the remainder of the division of $a$ by $N$}.
\ee
For example
$$
  7\mod5=2,\quad -3\mod5=2.
$$

We will also need \textit{congruence modulo} $N$:
\be\label{def-md}
  a\equiv b\md N~~\Leftrightarrow~~a\mod N=b\mod N\text{~~means $a$ and $b$ are congruent modulo $N$}.
\ee

In the formulating of the results (with a fixed $N$), we can avoid division with remainder and congruence, simply assuming that all actions are performed in the ring $\mathbb{Z}(N)$.
These operations are useful in constructing a proof by mathematical induction with different values of $N$.

The vacuum energy is the sum of the zero oscillation energies for all permissible values of the momentum and has the form
$$
  \mathcal{E}_{vac~b}
  =\mathcal{E}
  =\sum_{\mathbf{p}\in\mathbb{Z}^d(N),\mathbf{p}^2\in D}
   E(\mathbf{p}^2)\in\mathbb{Z}(N).
$$

Let us also consider a fermionic quantum field theory on the lattice $\mathbb{Z}^d(N)$ with a Hamiltonian of the form
\bea\label{H_f}
&&\hat H_{f}=\sum_{\mathbf{p}\in\mathbb{Z}^d(N),\mathbf{p}^2\in D}E(\mathbf{p}^2)\,\left(\hat a_{\mathbf{p}+}^\dagger\hat a_{\mathbf{p}+}-\hat a_{\mathbf{p}-}^\dagger\hat a_{\mathbf{p}-}\right),\\
\nonumber
&&
\mathbf{p}^2=\sum_{k=1}^d p_k^2\in\mathbb{Z}(N),\quad
E:D\to\mathbb{Z}(N),\quad 
D\subset\mathbb{Z}(N),\quad
[\hat a_{\mathbf{p_1}\sigma_1},\hat a_{\mathbf{p_2}\sigma_2}^\dagger]_+
=\ddelta_{\mathbf{p_1}\mathbf{p_2}}\,
\ddelta_{\sigma_1\sigma_2}\,\hat 1.
\eea
The vacuum energy is the sum of the negative energies of all the fermions of the Dirac Sea (for all permissible values of momentum) and has the form
\be
  \mathcal{E}_{vac~f}
  =-\mathcal{E}
  =-\sum_{\mathbf{p}\in \mathbb{Z}^d(N),\mathbf{p}^2\in D}
    E(\mathbf{p}^2)\in\mathbb{Z}(N).
\ee

For both bosons and fermions, we have limited ourselves to the case of one polarization, the case of an arbitrary number of polarizations can be considered completely similarly.

Since $\mathbf{p}^2\in\mathbb{Z}(N)$, the value of $\mathcal{E}$ can be calculated as follows
\be
  \mathcal{E}=\sum_{k\in D\subset\mathbb{Z}(N)}c(k)\,E(k)\in\mathbb{Z}(N),
\ee
The multiplicity $c(k)$ is
\be\label{c(k)}
  c(k)=c_{Nd}(k)=\Big(\text{the number of nodes $\mathbf{p}$ of the lattice }\mathbb{Z}^d(N)\text{ such that }
  \mathbf{p}^2\equiv k \md N\Big).
\ee
~\\

\textbf{Theorem.} For an arbitrary $N$ with $d\geqslant 3$, and for $N=2^n$ with $d\geqslant2$
\be\label{threv}
   \forall k\in\mathbb{Z}(N)\quad c_{Nd}(k)\equiv0 \md N.
\ee
The proof of the Theorem is given in the Appendix.

Note that if the theorem is proved for some $N$ and $d$, then for a given $N$ for all $d_1>d$, the theorem also holds.
~\\

When the conditions of the Theorem are fulfilled (for an arbitrary $N$ at $d\geqslant 3$, and for $N=2^n$ at $d\geqslant 2$) for an arbitrary function $E:D\to\mathbb {Z}(N)$ the vacuum energy calculated in the ring of residue classes  $\mathbb{Z}(N)$ is zero:
\be
  \mathcal{E}\equiv 0\md N.
\ee

The Hamiltonians \eqref{H_b} and \eqref{H_f}, for which the vacuum energy renormalization obtained, look like Hamiltonians for free fields. 
However, they can be considered as diagonalized Hamiltonians of the theory with interaction.
In this case, the function $E(\mathbf{p}^2)$ can be of arbitrary form, for example: $\mathbf{p}^2$ (field of non-relativistic particles),
$\sqrt{\mathbf{p}^2+m^2}$ (the field of relativistic particles), or even $(\mathbf{p}^2)^2$. 
If the theory is invariant with respect to shifts, then the momentum $\mathbf{p}$ will number possible excitations. If in the limit of large $N$ the theory becomes isotropic, the excitation energy can depend only on $\mathbf{p}^2$. Thus, vacuum energy renormalization takes place for a wide class of physically meaningful Hamiltonians.

More general Hamiltonian with zero vaccum energy in field $\mathbb{Z}(N)$ has the form
\be\label{general-H}
\hat H_g=\sum_{\mathbf{p}\in\mathbb{Z}^d(N),\mathbf{p}^2\in D}\hat h_\mathbf{p},\qquad
d\geqslant 3\text{ or }d=2,~N=2^2,
\ee
where the spectra of Hamiltonians $\hat h_\mathbf{p}$ has the form
\be
  E: (\mathbf{p}^2,\sigma)\mapsto E(\mathbf{p}^2,\sigma)\in\mathbb{Z}(N),
\ee
$\sigma\in\Sigma(k)$ is multi-index, which numerate energy levels for fixed $\mathbf{p}^2=k\in D\subset\mathbb{Z}(N)$.
To define vacuum energy field one has to choose an arbitrary linear order in the ring $\mathbb{Z}(N)$.
\be
  \mathcal{E}_{vac}=\sum_{k\in D\subset\mathbb{Z}(N)}c(k)\sum_{\sigma_0\in\Sigma_0(k)\subset\Sigma(k)}E(k,\sigma_0)
  \equiv0\md{N},
\ee
where $\Sigma_0(k)$ is defined by following condition
$$
E(k,\sigma_0)|_{\sigma_0\in\Sigma_0(k)}=\min_{\sigma\in\Sigma(k)} E(k,\sigma),\text{ or }
E(k,\sigma_0)|_{\sigma_0\in\Sigma_0(k)}<0.
$$

The Hamiltonians of  the form \eqref{general-H} describe a wide range of interacting lattice QFTs.

\section{Conclusion}

In comparison with the standard lattice QFT, what is new in our approach is that not only the momentum components, but also the energy run through discrete values on the lattice, and both the momentum projections and the energy are considered as elements of the ring  of residue classes $\mathbb{Z}(N)$. This automatically implies that spatial coordinates and time take values on the inverse lattice $\Delta x\cdot\mathbb{Z}(N)$ ($x\sim\Delta x\cdot N$).
The quantity that we call the momentum (the spatial shift generator on the lattice) in physics is often called \textit{quasi-momentum} (the quasimomentum coincides with the momentum in the continuous limit). Similarly, the quantity that we call energy (the generator of the time shift on the lattice) can be called \textit{quasi-energy}.

If we assume that <<Time is that which is measured by a clock.>> \cite{Bondi}, then in our model the clock is such that the time is given by a finite number of $q$-nary digits. After the clock counts down the maximum possible time for a given number of digits, the countdown begins from the beginning.

Since on a three-dimensional lattice, the quasi-energy of the vacuum is renormalized to zero for an arbitrary lattice size, the question arises whether it is possible to use the limit transition $N\to\infty$ to move from the lattice QFT to the QFT in continuous space. Is it possible to transfer the constructed number-theoretic renormalization to the continuous case?
There are reasons to believe that such a transfer is possible. In a series of papers \cite{bin-ivanov}, \cite{ter-bin}, \cite{digits}, the representation of the coordinate and momentum using a sequence of digits in an arbitrary number system was considered. Renormalization on the lattice was described as a change in the representation of $\mathbb{Z}(N)$: the transition from the representation of $\{0,1,\dots,N-1\}$ to the representation of $\{-n,-n+1,\dots,0,1,\dots,N-n-1\}$.
This lattice renormalization gives natural generalizations for the continuous limit at which the non-renormalized series passes into the renormalized one. The simplest of such generalizations (on a lattice it only works for $q=2$) is
$$
  x=\sum_{s=-\infty}^{+\infty} \mathbf{d}(s,x)\,q^s
  \quad\longrightarrow\quad
  x'=\frac1{q-1}\sum_{s=-\infty}^{+\infty} [\mathbf{d}(s-1,x)-\mathbf{d}(s,x)]\,q^s.
$$
Here $\mathbf{d}(s,x)$ is the digit in the position $s$ of the number $x$ in the $q$-nary number system.
The digit is assumed to be a periodic function of $x$
$$
  \mathbf{d}(s,x)=\mathbf{d}(s,x+q^{s+1}).
$$
The non-renormalized series for $x$ for some numeral systems may diverge for negative numbers (if numbers from the set $0,1,\dots,q-1$ are used) or for positive numbers (if numbers from the set $0,-1,\dots,-q+1$ are used) or for all numbers (if a set of digits not containing 0 is used).
The renormalized series converges in all these cases.

We are not yet ready to analytically calculate the individual digits of the $q$-nary expansion of integrals for vacuum energy, for this reason we started with the study of finite lattices.

The fact that the quasi-energy of the vacuum turned out to be zero for any lattice size, as well as the fact that the number-theoretic renormalization was previously determined not only for the finite lattice, but also on the real line, seems to be a weighty argument in favor of the physical meaningfulness of the transition to the continuous limit $N\to\infty$.

The direct calculation of the number-theoretic renormalization in the continuous case may be facilitated by the previously found \cite{bin-ivanov}, \cite{ter-bin} integral representation of real numbers and number-theoretic renormalizations in such a representation:
$$
x=\int\limits_{s=-\infty}^{+\infty} \mathbf{d}(s,x)\,q^s\,ds
\quad\longrightarrow\quad
x'=\frac1{q-1}\int\limits_{s=-\infty}^{+\infty} [\mathbf{d}(s-1,x)-\mathbf{d}(s,x)]\,q^s\,ds.
$$
Here the digit position $s$ runs through all real (including fractional) values
$$
  \mathbf{d}(s,x)=\mathbf{d}(0,x\cdot q^{-s}).
$$
The other form of number-theoretic renormalization for the integral representation is obtained by integration in parts with omission of boundary terms:
$$
x=\int\limits_{s=-\infty}^{+\infty} \mathbf{d}(s,x)\,q^s\,ds
\quad\longrightarrow\quad
x''=-\frac1{\ln q}\int\limits_{s=-\infty}^{+\infty} \frac{d(\mathbf{d}(s,x))}{ds}\,q^s\,ds.
$$
The function $\mathbf{d}(s,x)$ experiences jumps (it is piecewise constant between jumps), so the derivative
$\frac{d(\mathbf{d}(s,x))}{ds}$ is understood in the sense of generalized functions.

The integral representation of real numbers is overdetermined (it is enough to know any series of values of $\mathbf{d}(s,x)$ with a unit step of $s$), but there are no special scales in the integral representation (in the representation in the form of a series of $s\in\mathbb{Z}$, scales of the form $q^s$ are special).

Above we discussed the possibility of switching from $\mathbb{Z}(N)$ to real numbers, however, we should also consider the possibility of switching to $p$-adic numbers \cite{p-adic-vvz}. If the lattice size is set as $N=p^m$, then at $m\to\infty$ the vacuum energy tends to zero in the $p$-adic sense. This point of view on the renormalization of vacuum energy also needs further investigation.

The problems under study are also interesting from a number-theoretic point of view.
In the course of the research, one of us (V.N.) formulated and numerically tested the following hypothesis:

\textit{1) For arbitrary integers $m\geqslant0$ and $N\geqslant1$, the number of ways in which an arbitrary element $X$ of the ring $\mathbb{Z}(N)$ is represented as the sum of $d$ terms of degree $m$
$$
  X\equiv x_1^m+\dots+x_d^m \md{N}
$$ 
is always a multiple of $N$ if $d\geqslant m+1$.}

\textit{2) If $d<m+1$, then there is such a $N$, and such a $X\in\mathbb{Z}(N)$, that the number of ways in which $X$ is represented as the sum of $d$ terms of degree $m$ is not divisible by $N$.
}

The hypothesis is obvious for $m=0$ and $m=1$. For $m=2$, the hypothesis follows from the Theorem in this paper. For higher degrees, the hypothesis has so far been tested numerically (V.N.).

The first statement of the hypothesis is verified numerically for all cases
\begin{itemize}
	\item $m\leqslant8$, $N\leqslant1000$;
	\item $m\leqslant35$, $N\leqslant300$;
	\item $m\leqslant100$, $N\leqslant37$.
\end{itemize}
The second statement of the hypothesis is verified for all $m\leqslant 100$.

\subsection*{Acknowledgement}

The authors thank the participants of the seminar of the Department of Mathematical Physics of Steklov Mathematical Institute of Russian Academy of Sciences and
the seminar of the Laboratory of Infinite Dimensional Analysis and Mathematical Physics of the Faculty of Mechanics and Mathematics of Moscow State University
for valuable comments and discussion. 
We are especially grateful for the fruitful and friendly discussion
of I.V. Volovich, E.I. Zelenov, N.N. Shamarov, S.L. Ogarkov, Z.V. Khaydukov, D.I. Korotkov, V.V. Dotsenko.

\section*{Appendix}
\addcontentsline{toc}{section}{Appendix}

\appendix
\section{Quadratic congruences} \label{AppA}
\subsection{Quadratic congruences to prime modulus} 

Simplest quadratic congruence 
$$x^2 \equiv 0 \md{p}$$
in field $\mathbb{Z}(p)$ means that square root of 0 exists in $\mathbb{Z}(p)$ and defined uniquely.

Quadratic congruence  
\begin{equation}\label{mod}
	x^2 \equiv a \md{p}, \qquad\text{ in case } {a}\not\equiv0 \md{p}
\end{equation}
has solutions for some $a\in \mathbb{Z}(p)$, and such $a$ is said to be a quadratic residue. Note that the
trivial case  $a=0$ is generally excluded from set of quadratic residues $\text{Res}(p)= \{ 1, r,s,\dots\} $, which are congruent to integers 
$$
1^2, 2^2, ..., \left( \frac{p-1}{2}\right)^2. 
$$
The number of quadratic residues is exactly $\frac{p-1}{2}$ because  if we take different $x_1$ and $x_2$
$$
x_1^2\equiv x_2^2\equiv a\not\equiv0 \md{p},
$$
then, since  $\mathbb{Z}(p)$ is field, we have
$$
(x_1/x_2)^2\equiv 1 \md{p}\quad\Rightarrow\quad x_1=\pm x_2.
$$
If the congruence  (\ref{mod}) does not have a solution, then $a$ is said to be a quadratic nonresidue.

The Legendre symbol is defined for integers $ a \in{\mathbb{Z}} $:
\begin{equation}\label{chi}
	\chi(a) \stackrel{\text{def}}{=}
	\quad\left(\frac{a}{p}\right)_L  
	\stackrel{\text{def}}{=} 
	\left\{ \begin{aligned}
		& \quad 0 \quad\text{ if $a$ is divisible by } p,\\
		& \quad 1 \quad\text{ if $a$ is quadratic residue modulo } p,\\
		& -1 \quad\,\text{if $a$ is quadratic nonresidue}.\\	
	\end{aligned}\right.  
\end{equation}

Legendre symbol is group theoretical character and we denote it like  $\chi(a)$ for $p$ fixed.
From Fermat's and Legendre' theorems it follows that
\begin{equation}\label{chia}
	\forall a : \qquad \chi(a) \equiv a^{\frac{p-1}{2}} \md p. 
\end{equation} 
Using this congruence one can simply evaluate Legendre symbol of $-1$  by formula
$$ 
\chi(-1) = (-1)^{\frac{p-1}{2}}.
$$
Also characters' sum  over all elements of $\mathbb{Z}(p)$ is equal to $0$:
\begin{equation}\label{sumchi}
	\sum_a \chi(a)  =0. 
\end{equation}
Hence (in one dimension and $N=p$) the multiplicity 
(\ref{c(k)}) as number of solutions of congruence (\ref{mod}) is of the form
$$
c
(k)=1+ \chi(k),\quad
k\in\mathbb{Z}(p).
$$

\newpage
\subsection{Quadratic congruence to $p^m$ modulus}
We can consider all solutions $x$ of  congruence with parameter $a$
\begin{equation}\label{modpn}
	x^2\equiv a \md{p^m},\quad 
	x,\, a \in \mathbb{Z}(p^m).
\end{equation}
Any solution of (\ref{modpn}) will be also the solution of similar congruence with  less by 1 power of $p$:
\begin{equation}\label{modpnn1}
	\text{if }\acute{x}:\; \acute{x}^2\equiv a\md{p^m} \text{ then }
	\;\Rightarrow\; 
	\acute{x}^2\equiv a\md{p^{m-1}} \;\Rightarrow\;\dots\Rightarrow\;
	\acute{x}^2\equiv a\md{p^1}.
\end{equation}
Notation of $x$ expansion in powers of $p^\alpha$ and digits $x_\alpha$  looks like:
\begin{equation} \label{x&a}
	\begin{aligned}
		x  &= x_0 +x_1\, p + x_2\, p^2 + \cdots + x_\alpha p^\alpha + \cdots +x_{m-1}\, p^{m-1},\quad 
		x_\alpha \in \mathbb{Z}(p) ;\\
		a  &= a_0 +a_1\, p + a_2\, p^2 + \cdots + a_\alpha p^\alpha + \cdots  +a_{m-1}\, p^{m-1},\quad  
		a_\alpha \in \mathbb{Z}(p).\\
	\end{aligned} 
\end{equation}
First we explain the case when $a$ relatively prime to $p$. 

\paragraph*{For $a$ with digit $a_0\neq0$.} 
The greatest common divisor of $a$ and $p$ in this case is 1  
$$
(a,p)=1.
$$
From (\ref{modpnn1}) and (\ref{x&a}) one can immediately find  congruence for digit $x_0$
\begin{equation}\label{modpn1}	
	x_0^2 \equiv a_0 \md{p^1},
\end{equation}
it looks like (\ref{mod}). If $a_0$ is a quadratic nonresidue there are no solutions of (\ref{modpn}). 
So one-dimensional multiplicity for $N=p^m$
\begin{equation*}
	c
	(a)=0 \quad
	\text{if $(a,p)=1$ and digit $a_0 \not\in\text{Res}(p)$.}
\end{equation*}

We go further with $a_0$ as quadratic residue,
so $\; x_0=\pm\sqrt{a_0}\in \mathbb{Z}(p)$.
From (\ref{modpnn1}) and (\ref{x&a})  from next $p^2$-modulo congruence $x^2\equiv a\md{p^2}\;$ we get 
\begin{equation}\label{modpn2}
	x_0^2 + 2 x_0 x_1 p +\cancel{x_1\, p^2}\equiv a_0+a_1 p \md{p^2}
\end{equation}
as linear congruence with respect to unknown $x_1$
$$
2 x_0  x_1 \; p\equiv a_0 -x_0^2 +a_1 p \md{p^2},
$$
one can reduce this right  side by $p$ due to (\ref{modpn1}). Then
$$
2 x_0 \; x_1 \equiv a_1 + \frac{a_0 -x_0^2}{p}  \md p.
$$
In field $\mathbb{Z}(p)$ since  $2x_0\neq0$ we have inverse element $(2 x_0)^{-1}\in \mathbb{Z}(p)$.
Digit $x_1$ is uniquely defined:
$$
x_1 \equiv \left(a_1 +  \frac{a_0 -x_0^2}{p} \right) (2 x_0)^{-1} \md p.
$$

From next $p^3$-modulo congruence $x^2\equiv a\md{p^3}$ we get linear congruence with unknown $x_2$
$$
\begin{aligned}
	(x_0 +x_1\, p)^2 + 2 (x_0 +\cancel{x_1\, p} ) x_2\, p^2 &\equiv a_0+a_1 p + a_2\, p^2\md{p^3} \quad \Rightarrow	\\
	2 x_0\,  x_2 \; p^2 &\equiv a_0+a_1 p -(x_0 +x_1\, p)^2 +a_2 p^2 \md{p^3},
\end{aligned}
$$
which right  side has common factor  $p^2$ due to (\ref{modpn2}). 
So at this step one can find  $x_2$ uniquely:
$$
x_2 \equiv \left( a_2+ \frac{a_0+a_1 p -(x_0 +x_1\, p)^2}{p^2} \right) (2 x_0)^{-1} \md p.
$$

All digits $x_\alpha$ are  found sequentially step by step. 
Let us consider arbitrary  step of finding $x_\alpha$ from $p^{\alpha+1}$-modulo congruence $x^2\equiv a\md{p^{\alpha+1}}$. 
$$
\begin{aligned}
	([x_0 +x_1\, p +\cdots]+ x_\alpha p^\alpha)^2 &\equiv [a_0+a_1 p +\cdots]+ a_\alpha p^\alpha \md{p^{\alpha+1}};\\
	[x_0 +\cdots]^2 + 2 [x_0 +\cancel{\cdots} ] x_\alpha\, p^\alpha &\equiv [a_0+a_1 p +\cdots]+ \; a_\alpha p^\alpha \md{p^{\alpha+1}}; \\
	2 x_0\,  x_\alpha\, p^\alpha &\equiv [a_0 +\cdots]- [x_0 +\cdots]^2 +a_\alpha p^\alpha \md{p^{\alpha+1}}; \\
\end{aligned}
$$
We have $x_\alpha \in\mathbb{Z}(p) $ uniquely as
$$
x_\alpha \equiv \left(a_\alpha+  \frac{[a_0+\cdots] -[x_0 +\cdots]^2}{p^\alpha}  \right) (2 x_0)^{-1} \md p,
$$
when $x_0$ is fixed with its sign, then  solution $x$ is calculated uniquely.


Thus multiplicities in one dimension
\begin{equation}\label{cpma0}
	c_{p^m}(a) =c_{p}(a_0) =1 +\chi(a_0)\quad
	\text{if digit } a_0\ne0,
\end{equation}
and we obtain this property of multiplicities for arbitrary $n<m$
\begin{equation}\label{cpmcpn0}
	c_{p^m}(a)=c_{p^n}\left( a{\mod p^n}\right) \quad
	\text{if digit }a_0\ne0.
\end{equation}   

\paragraph*{For $a$ with $a_0=0$ and $a_1 \neq0$.}
From  $p^2$-modulo congruence $x^2\equiv a\md{p^2}\;$ we get
$$
(x_0+ x_1\, p + \cancel{x_2\, p^2} +\cdots)^2  \equiv a_1 p \md{p^2}.
$$
and anyway no solutions for $x_1$. So if $a_1 \neq0$ than $c_{p^m}(a)=0$, i.e.
\begin{equation*}
	c_{p^m}(\xi\, p^1)=0, \quad
	\xi\in\mathbb{Z}(p^{m-1})\setminus\{0\}.
\end{equation*}
\paragraph*{For $a$ with $a_0=a_1=\cdots =a_{2\lambda}=0$ and $a_{2\lambda+1}\neq0$, $\lambda\geqslant0$.} 
We put the lowest non-zero digit in $x$  as $x_\alpha\neq0$. Then congruence \eqref{modpn} takes the form
\begin{equation*}
	(x_\alpha\cdot p^\alpha+\cdots)^2 
	\equiv (a_{2\lambda+1}+\cdots)p^{2\lambda+1} \md{p^m}.
\end{equation*}
\begin{equation*}
	(x_\alpha+\cdots)^2 p^{2\alpha} 
	\equiv (a_{2\lambda+1}+\cdots)p^{2\lambda+1} \md{p^m}.
\end{equation*}
anyway no solutions exist and multiplicity
\begin{equation}\label{cpm0a0a1a}
	c_{p^m}(\xi\, p^{{2\lambda+1}})=0, \quad
	\text{where } \xi\in \mathbb{Z}(p^{m-2\lambda-1}) \setminus\{0\},\, \lambda\geqslant0.
\end{equation}

\paragraph*{For $a$ with $a_0=a_1=0$ and $a_2 \neq0$.}
From  $p^2$-modulo congruence $x^2\equiv a\md{p^2}\;$ we get
$$
(x_0+ x_1\, p + \cancel{x_2\, p^2 +\cdots})^2  \equiv 0 \md{p^2}.
$$
Simply $x_0=0$ and from  $p^3$-modulo congruence $x^2\equiv a\md{p^3}\;$ we get ordinary quadratic congruence for $x_1$ 
$$
( x_1\, p + \cancel{x_2\, p^2+\cdots} )^2  \equiv a_2 p^2 \md{p^3} \;\Rightarrow\; x_1^2 \equiv a_2 \md{p},
$$
with some solutions if $\sqrt{a_2}$ exists.
But now starting equation is not  (\ref{modpn}), because of power descending by 2 units
$$
( x_1\, p + x_2\, p^2+\cdots +x_{m-1}\, p^{m-1})^2  \equiv a_2 p^2+\cdots+a_{m-1}\, p^{m-1} \md{p^m} 
$$
$$
\;\Rightarrow\; ( x_1\, + {x_2\, p^1+\cdots} +\cancel{x_{m-1}\, p^{m-2}})^2  \equiv a_2 +\cdots+x_{m-3}\, p^{m-1} \md{p^{m-2}},
$$
thus obviously last digit $x_{m-1}$ can take arbitrary values from $\mathbb{Z}(p)$.
Therefore multiplicities as number of solutions of  (\ref{modpn}) is greater by $p$
$$
c_{p^m}(a)=\bigl(1+\chi(a_2) \bigr)\, p\quad
\text{if digit } a_2\ne0,
$$
i.e.
\begin{equation*}
	c_{p^m}(\xi\, p^2)=\bigl(1+\chi(\xi{\mod p}) \bigr)\, p, \quad \text{where }
	\xi\in\mathbb{Z}(p^{m-2})\setminus\{0\}.
\end{equation*}

\paragraph*{For $a$ with $a_0=a_1=\cdots =a_{2\lambda-1}=0$ and $a_{2\lambda}\neq0$, $2\lambda<m$.} 
\begin{equation*}
	(x_\alpha+\cdots)^2 p^{2\alpha} 
	\equiv (a_{2\lambda}+\cdots)p^{2\lambda} \md{p^m}, \quad
	x_\alpha\neq0.
\end{equation*}
The solutions exist only when ${\alpha}=\lambda$. 
We obtain
\begin{equation}\label{x1}
	(x_\lambda+\cdots+\cancel{x_{m-1}\, p^{m-1-\lambda}})^2 \equiv a_{2\lambda} +\cdots \md{p^{m-2\lambda}},
\end{equation}
and we have the similar case as (\ref{cpma0})  (when $a$ is relatively prime to $p$).

in solution $x$ we see that the lowest $\alpha$ digits are zero and the highest $\alpha$ digits you can take arbitrary. And between them $m-2\alpha$ digits are calculated step by step from linear congruences.

Therefore multiplicities as number of solutions of  (\ref{modpn}) is greater by $p^{\alpha}$
\begin{equation}\label{fpma2k}
	c_{p^m}(a)= \bigl(1+ \chi(a_{2\alpha}) \bigr) p^{\alpha}, \quad 
	\text{if digits  $a_0=\cdots =a_{2\alpha-1}=0$, $a_{2\alpha}\neq0$ and } 2\alpha<m.
\end{equation}
And for arbitrary $n<m$
\begin{equation}\label{cpmcpn2}
	c_{p^m}(a)=c_{p^n}\left( a{\mod p^n}\right)  \quad
	\text{if digit $a_{2\alpha}\ne0$ and $2\alpha\leqslant n<m$.}
\end{equation} 
In section \ref{AppNpm} we'll investigate similarity of generating functions in $d=1$ using this relation.

\paragraph*{For $a=0$ with $a_0=\cdots =a_{m-1}=0$.} We denote the lowest non-zero digit in $x$ by index $\alpha$.
\begin{equation}\label{x2_0pm}
	(x_\alpha p^\alpha+\cdots)^2=(x_\alpha+\cdots+x_{m-1}\, p^{m-1-\alpha})^2 \; p^{2\alpha} \equiv0 \md{p^m}.
\end{equation}
We have solutions here if
$$
2\alpha \geqslant m,
$$
so the lowest value of index $\alpha$ is $
{\left\lfloor\frac{m+1}{2}\right\rfloor}.
$
The number of arbitrary digits is $\left\lfloor\frac{m}{2}\right\rfloor$.  
So multiplicity is $c_{p^m}(0)=p^{\lfloor\frac{m}{2}\rfloor}$.
We get all-for-one in formul\ae{} 
\begin{equation}\label{cpm}
	c_{p^m}(a) 
	=\left\{ \begin{aligned}
		p^{\lfloor\frac{m}{2}\rfloor} 
		\quad 
		&\text{ if }a=0, \\
		\bigl(1+ \chi(a_{2\lambda}) \bigr)\, p^{\lambda} \quad 
		&\text{ if $a$  is divisible by $p^{2\lambda}$, i.e. $(a,p^m)=p^{2\lambda}$, and  $\lambda\geqslant0$,}\\	
		0 \quad 
		&\text{ if }a\text{ is divisible by } p^{2\lambda+1},\text{ i.e. }(a,p^m)=p^{2\lambda+1}, \lambda\geqslant0. \\
	\end{aligned}\right.  
\end{equation}

\section{Generating  functions for multiplicities }\label{Appendix0}
We define the polynomial of $\tau$ 
\be\label{fdN}
f^d_N(\tau)=\sum_{k\in\mathbb{Z}(N)} c_{Nd}(k)\,\tau^k
\ee
as a generating  function for multiplicities $ c_{N\,d}(k) $ of $E(k)$ on the lattice $\mathbb{Z}^d(N)$.
For  one-dimensional generating function we  skip the upper index $ d=1 $, i.e. $f_N(\tau)\stackrel{\text{def}}{=}f^1_N(\tau)$. And $f_0(\tau)\stackrel{\text{def}}{=}1$.

We use brackets
\footnote{$\left\langle \tau^{k} \mid f_N(\tau) \right\rangle =\frac{1}{N} \sum_{l=1}^N \tau(l)^* f_N(\tau(l)) $, where $\tau(l)$ are roots of unity in $\mathbb{C}$ given by $\tau(l)=\E^{2\pi \I \, l/N} $. }
to extract multiplicity $ c_{N\,d}(k) $ from generating  function: 
\begin{equation}\label{bra}
	\left\langle \tau^{k} \mid f^d_N(\tau) \right\rangle =c_{N\,d}(k).
\end{equation}

We treat $\tau$ as a formal variable with the equivalence relation:
\be\label{tauN}
\tau^N=1,
\ee
i.e. the powers of $\tau$ in (\ref{fdN}) can be considered as elements of $\mathbb{Z}(N)$.
The multiplication with  relation (\ref{tauN}) we denote as $f\circ g$.
Hereafter $\tau^N$ is equivalent to unity in reduction of similar terms so that $\tau^{k+N}=\tau^k$:
$$
\left(\sum_{p_1\in\mathbb{Z}(N)} \tau^{p_1^2}\right)^d
=\sum_{p_1\in\mathbb{Z}(N)}\cdots \sum_{p_d\in\mathbb{Z}(N)} \tau^{p_1^2}\circ \cdots \circ\tau^{p_d^2}
= \sum_{\mathbf{p}\in\mathbb{Z}^d(N)} \tau^{\mathbf{p}^2}
,
$$ 
hence in case $d>1$ the generating function can be represented as a $d$ times product of one-dimensional generating functions:
\be\label{f^d}
f^d_N(\tau)
=\Bigl(f_N(\tau)\Bigr)^{ d}.
\ee

The introduced equivalence (\ref{tauN}) doesn't violate the commutativity and distributivity of multiplication. Thereby, we reduce our investigation of multiplicities (\ref{c(k)}) to problems of polynomial algebra and arithmetic of quadratic forms \cite{dotsenko}, \cite{vavilov}.

One can assume that variable $\tau$ takes values in the set of complex $N$th roots of unity:
$$
\tau\in\{\tau(l)=\E^{2\pi\I\frac{l}{N}},\, l\in\mathbb{Z}(N)\}.
$$  In such a representation we find multiplicities $c_{Nd}(k)$ as coefficients of  discrete Fourier expansion of the function $f^d_N(\tau(l))$ on the lattice  $\mathbb{Z}(N)$.

\subsubsection*{Polynomial  $\phi_N(\tau)$ }  
We  define the polynomial $\phi_N(\tau)$ with all coefficients  equal to 1: 
\begin{equation}\label{keyF}
	\phi_N(\tau)=\sum_{k\in\mathbb{Z}(N)} \tau^k=  1 +  \tau +  \tau^2 + \cdots +  \tau^{N-1}.
\end{equation} 
Here $\phi_N(\tau) \not\equiv \frac{1-\tau^N}{1-\tau}$ because due to (\ref{tauN}) the  value 1 of variable $\tau$ is possible value.

If $\tau$ runs over complex roots of unity then 
$$
\phi_N(\tau)~=~\left\{\begin{array}{cc}
	N &\text{if }\tau=1,\\
	0 &\text{otherwise}.
\end{array}\right.
$$
From this point of view $\phi_N(\tau)$ is the lattice analogue of the delta function $$
\phi_N\bigr(\tau(l)\bigl)=\left\{\begin{array}{cc}
	N &\text{if }l=0,\\
	0 &\text{otherwise}
	.
\end{array}\right.$$

For any  $k\in\mathbb{Z}(N)$ we have
\be\label{tauphi} 
\tau^k \circ \phi_N \; (\tau)= 1\cdot\phi_N (\tau) ,
\ee
hence
\be\label{phi2}
\phi_N \circ \phi_N = \left( \sum_{k=0}^{N-1} \tau^{k}\right) \circ \phi_N = N\cdot \phi_N 
\ee
and consequently
\be\label{NminusphiN}
\phi_N \circ \left(N- \phi_N\right) = 0,
\ee
\begin{equation*}
	\begin{aligned}
		\left(N- \phi_N\right)^{\;2}
		&= N\left(N- \phi_N\right) ,\\
		\left(N- \phi_N\right)^{\;m}
		&= N^{m-1}\left(N- \phi_N\right).
\end{aligned}\end{equation*}

\section{ Case $N=p$, $p\not=2$}\label{N=p}
Now $p$ is prime odd number and we put $N=p$ hereafter in section \ref{Appendix0}. During intermediate calculations we sometimes 
skip the index $p$ (or $N$).

For the one-dimensional case we can derive the explicit expression for the multiplicities $ c_{N\,1}(k) =c_{N}(k)$ in terms of  quadratic congruences \cite{vinogradov}.

\subsubsection*{Polynomial  $g_p(\tau)$ }  
Coefficients $c_{p}(k)$ take values $0$ or $2$, or $1$ (square root of $0$); and the generating function is
$$ 
f_p(\tau)= 1 + 2 \tau+ 2 \tau^r + 2 \tau^s + \cdots =1+\sum_{r\in \text{Res}(p)}2\tau^r,
$$ 
where powers $\tau$ with nonresidues $\tilde{r},\tilde{s},..  $ disappeared from $f_p(\tau)$ because of its null coefficients (multiplicities).

Next we introduce the polynomial $ g_p(\tau)$ using the character $\chi(k)$ (defined in section \ref{AppA}):
\begin{equation}\label{gpchi}
	g_p(\tau)= \sum_{k\in\mathbb{Z}(p)} \chi(k)\cdot \tau^k
	. 
\end{equation}

So, the generating function for multiplicities in the one-dimensional case is   the sum of polynomials $g_p(\tau)$ и $\phi_p(\tau)$:
\begin{equation}\label{f=fi+g}
	f_p(\tau)= 1+\sum_{r\in \text{Res}(p)}2\tau^r 
	= \phi_p(\tau) +g_p(\tau). 
\end{equation}
And multiplicities in $d=1$ are given by explicit formula:
$$
\text{in $d=1$ } \quad c_{p\,1}(k)=1+\chi(k).
$$

\subsection{ Properties of polynomials $g_p$ and $\phi_p$ in multiplication}
Due to \eqref{phi2} we have
\be\label{fipfip}
\begin{aligned}
	\phi_p &\circ \phi_p = p\cdot \phi_p \quad\text{ and }\quad 
	\phi_p  ^{\; m} = p^{m-1}\cdot \phi_p,\\
	(p - \phi_p) &\circ \phi_p = 0,
\end{aligned}
\ee
and   we get
\be\label{gphi}
\phi_p \circ g_p =0,
\ee
because of \eqref{tauphi} and \eqref{sumchi} 
\begin{equation*}
	\phi_p \circ \left( \sum_{a=0}^{p-1}  \chi(a)\,\tau^a \right) 
	=\phi_p \circ \Biggl( \underbrace{\sum_{a=0}^{p-1}  \chi(a) }_0 \Biggr) 
	=0.
\end{equation*}

\subsubsection*{ Proposition for $g_p^{\,2}$. } 
The explicit result for $g_p^{\,2}$ is:
\begin{equation}\label{g2_}
	g_p\circ g_p 
	=  (-1)^{\frac{p-1}{2}} \left(\vphantom{\frac{a}{a}} p- \phi_p
	\right).
\end{equation} 

\begin{proof}
	Writing out $g_p^{\,2}$ as a double sum and remembering about \eqref{tauN}, we will make a transition at  fixed $a$ from  summation variable $b$ to the new summation variable 
	$t$ by simple rule \mbox{$b\equiv a\cdot t \md p$}, since  $\mathbb{Z}(p)$ is field. In  third line, we split the sum over $t$ into two parts.
	\begin{eqnarray*}
		g_p\circ g_p \;(\tau) = \left( \sum_{a} \chi(a)\tau^a\right)  \circ \left( \sum_{b} \chi(b)\tau^b\right) 
		= \sum_{a,b=1}^{p-1} \chi(a)\chi(b)\tau^{a+b}= \\
		= \sum_{a,t=1}^{p-1} \chi(a)\chi(at)\tau^{a+at }
		= \sum_{t=1}^{p-1}\sum_{a=1}^{p-1} \cancelto{1}{\chi^2(a)}\chi(t)\tau^{a(1+t) } =\\
		=\chi(t)\big|_{t=p-1}\sum_{a=1}^{p-1} {\tau^{a\cdot0}} +
		\sum_{t=1}^{p-2}\chi(t) \sum_{a=1}^{p-1} \tau^{a(1+t)}
		= \chi(-1)\sum_{a=1}^{p-1} 1 +\sum_{t\neq -1}\chi(t) \sum_{a=1}^{p-1} \tau^{a} =\\
		= \chi(-1)(p-1) + \left(\vphantom{\frac{a}{a}} 0-\chi(-1)\right)\cdot \bigl( -1+\phi_p(\tau)\bigr) 
		= \chi(-1)\left(\vphantom{\frac{a}{a}} p- \phi_p(\tau) \right). 
	\end{eqnarray*} 
\end{proof}
As a consequence of (\ref{fipfip}), (\ref{gphi}), (\ref{g2_}) we have the equation for arbitrary powers $m$ and $n$
\begin{equation}
	\phi_p^{\;m} \circ g_p^{\;n} =0,	
\end{equation}
then we obtain
\begin{equation}\label{phi+g_^m}
	(\phi_p + g_p) ^{\;m} = \phi_p^{\;m} + g_p^{\;m}.
\end{equation}
And from (\ref{g2_})
\begin{equation}\label{g3_}
	g_p\circ g_p \circ g_p 
	=  (-1)^{\frac{p-1}{2}} \cdot p\cdot g_p 
	.
\end{equation} 
Arbitrary power $m$:
\begin{equation}\label{g^m}
	g_p ^{\; m}=
	(-1)^{\frac{p-1}{2}\lfloor\frac{m}{2}\rfloor}
	\begin{cases}
		p^{\lfloor\frac{m}{2}\rfloor} \cdot g_p &\text{if odd $m$,}\\
		p^{\frac{m}{2}-1} \bigl(p -  \phi_p) &\text{if even $m$.}\\
	\end{cases}
\end{equation}

\subsection{Multiplicities in $d=2$ }
From generation function \eqref{f=fi+g} we have in two-dimensional case due to (\ref{phi+g_^m})
$$
f^2_p(\tau) = f_p \circ f_p
= \phi_p \circ \phi_p 
+ g_p\circ g_p.
$$
Now  we have mathematical expression for $f^2_p$ for any prime odd $p$:
\begin{equation}\label{f2p}
	f^2_p = p\cdot\phi_p
	+ (-1)^{\frac{p-1}{2}} \bigl( p- \phi_p
	\bigr),
\end{equation}
in explicit polynomial form it looks as
\begin{equation}\label{f2}
	f^2_p(\tau)=  
	\begin{cases}
		(2p-1)\cdot\tau^0  +  (p-1)\cdot\tau^1 +\cdots+(p-1)\cdot\tau^{p-1} &
		\text{if even $\frac{p-1}{2}$,}\\ 	  
		\qquad 1\cdot\tau^0	  \; + \; (p+1)\cdot\tau^1 +\cdots+(p+1)\cdot\tau^{p-1} &
		\text{if odd  $\frac{p-1}{2}$}.
	\end{cases}
\end{equation}
So polynomial coefficient from $f_p^2$ is congruent with $\pm1$ and multiplicity is congruent as
\begin{equation}\label{cp2k}
	\text{in $d=2$ } \quad c_{p\,2}(k)\equiv -\chi(-1) \md{p},\quad k\in\mathbb{Z}(p).
\end{equation}

\subsection{Multiplicities in $d=3$ }

Using (\ref{phi+g_^m}) and (\ref{g^m}) for $f^3_p$ we have
\begin{equation*}
	f^3_p = f_p \circ f_p \circ f_p = 
	\phi_p^{\;3} + g_p^{\;3} 
	= p^2\,\phi_p +(-1)^{\frac{p-1}{2}}\, p \, g_p.
\end{equation*}
In the expression above one can take the common factor $p$ out of brackets. So each coefficient in polynomial above is divisible by $ p $ and we have
$$ c_{p\,3}(k)\equiv 0 \md p .$$ 

And for $ d>3 $ formula
$$ c_{p\,d}(k)\equiv 0 \md p$$ 
is obvious due to factorization (\ref{f^d}). We proved the proposition of the Theorem in case $N=p$.

\section{Case $N=p^m$, $p\not=2$ } \label{AppNpm}

\subsubsection*{Cyclotomic Polynomials $\Phi_n(\tau)$}

In this section $N=p^m$,  $p$ is prime odd and  $m\in\mathbb{N}$, $m>1$.
The equivalence relation (\ref{tauN}) now is
\begin{equation*}
	\tau^{p^m}=1.
\end{equation*}


In addition to $\phi_N(\tau)$ as sum of all powers in (\ref{keyF}) we introduce cyclotomic polynomials for $N=p^m$:
\begin{equation}\label{Phipm=}
	\begin{aligned}
		\Phi_p (\tau) &= 1 +  \tau +  \tau^2 +  \cdots +  \tau^{p-1}=\phi_p (\tau) , \\
		\Phi_{p^2}(\tau) =\Phi_p(\tau^p) &=  1 +  \tau^{p}  +  \tau^{2p}+  \tau^{3p}+ \dots +\tau^{(p-1)p} =\phi_p(\tau^p), \\
		\Phi_{p^3}(\tau) =\Phi_p(\tau^{p^2}) &=  1 +\tau^{p^2} +\tau^{2p^2} +\tau^{3p^2} +\cdots+\tau^{(p-1)p^2} =\phi_p(\tau^{p^2}), \\
		\vdots\\	
		\Phi_{p^{m}}(\tau) =\Phi_p(\tau^{p^{m-1}}) &=  1 +\tau^{p^{m-1}} +\tau^{2p^{m-1}} +\tau^{3p^{m-1}} +\cdots+\tau^{(p-1)p^{m-1}} =\phi_p(\tau^{p^{m-1}}) .\\
	\end{aligned} 
\end{equation}

One can easily derive that the product of cyclotomic polynomials of successive powers of $p$ is the sum of all powers of $\tau$ in $\phi_N(\tau)$: 
$$
\phi_{p^{m}} = \Phi_{p}\cdot\Phi_{p^2}\cdots\Phi_{p^{m}}.
$$

We have similar to (\ref{fipfip}) equalities 
\be\label{Fpm}
\begin{aligned}
	\tau^{p^{m-1}} \circ \Phi_{p^{m}} 
	=  \Phi_{p^{m}}, \\
	\Phi_{p^{m}} \circ \Phi_{p^{m}}
	= p\cdot \Phi_{p^{m}},\\
	\left(p- \Phi_{p^{m}} \right)  \circ \Phi_{p^{m}} 
	= 0.
\end{aligned}
\ee
And polynomials with modified argument $\tau^{p^{m-1}}$
\be\label{phi*g_pm}
\Phi_{p^{m}} (\tau) \circ g_p (\tau^{p^{m-1}}) =
\phi_p (\tau^{p^{m-1}}) \circ g_p (\tau^{p^{m-1}}) =0,
\ee
\begin{equation}\label{pm_phi+g_^n}
	\bigl(\phi_p (\tau^{p^{m-1}}) + g_p (\tau^{p^{m-1}}) \bigr) ^{\;n} = \phi_p (\tau^{p^{m-1}}) ^{\;n} + g_p (\tau^{p^{m-1}}) ^{\;n}.
\end{equation}

\subsection{Similarity of successive generating functions in $d=1$}

\subsubsection*{ Proposition. } The recurrent formula for $f_{p^m}(\tau)=\mathcal{F}\left( f_{p^{m-1}}(\tau)\right) $  is
\begin{equation}\label{fpm=Phipm_fpm-1+g^m}
	f_{p^m}(\tau)  =
	\Phi_{p^m}(\tau) \cdot f_{p^{m-1}}(\tau) 
	+(-1)^{\frac{p-1}{2}\lfloor\frac{m}{2}\rfloor} \cdot
	\bigl(  g_p(\tau^{p^{m-1}}) \bigr)^m ,
\end{equation}
or via (\ref{g^m})
\begin{equation}\label{rrf2k1}
	f_{p^m}(\tau)  =
	\Phi_{p^m}(\tau)\cdot f_{p^{m-1}}(\tau) +  
	\begin{cases}
		p^{\lfloor\frac{m}{2}\rfloor} \cdot g_p(\tau^{p^{m-1}}) & \text{if odd } m,\\
		p^{\frac{m}{2}-1} \bigl(p -  \Phi_{p^m}(\tau)\bigr)& \text{if even } m.
	\end{cases}	
\end{equation}

\subsubsection*{The Proof for $f_{p^{m}}$ recurrent formula}
For example, if  we take  $k=1 \in\mathbb{Z}(p^{m})$ then from (\ref{cpma0}) in $d=1$ it follows that $c_{p^{m}}(1) = c_{p\,}(1)$. Using notation with brackets (\ref{bra}) we write 
\begin{equation*}
	\left\langle \tau^{1} \mid f_{p^{m}}(\tau) \right\rangle =\left\langle \tau^{1} \mid f_{p}(\tau) \right\rangle = 1+\chi(1) = 2
\end{equation*}
To generalize this relation we use (\ref{cpmcpn0}) for non divisible by $p$ arbitrary $k\in\mathbb{Z}(p^{m})$ and the equality of $\tau^k$ coefficients $c_{p^{m}}(k)$   in $f_{p^{m}}(\tau)$ and  $c_{p^{m-1}}(k {\scriptstyle \mod p^{m-1}} )$   in $f_{p^{m-1}}(\tau)$ is clear. Using cyclotomic polynomials (\ref{Phipm=}) we can write this similarity relation
\begin{equation}\label{tau^pow_nondiv_p}
	\Bigl\langle \tau^{k} \mid f_{p^{m}}(\tau)\Bigr\rangle
	= \Bigl\langle \tau^{k} \mid \Phi_{p^{m}}(\tau)\circ f_{p^{m-1}}(\tau)\Bigr\rangle
	= \Bigl\langle \tau^{k} \mid \Phi_{p^{m}} 
	\cdots\Phi_{p^{2}}\cdot f_p\Bigr\rangle
	\quad\text{ if } k\not\equiv0\md p .
\end{equation}
NB. In fact, the largest power of $\Phi_{p^{m}}$ polynomial is $(p-1)p^{m-1}$ and the largest power of $f_{p^{m-1}}$ polynomial is $p^{m-1} -1$ and power of their product is
$$
(p-1)p^{m-1} +p^{m-1} -1=p^{m}-1,
$$
and no need to use here $\circ$ product.

We can obtain formula more general  than (\ref{tau^pow_nondiv_p}). Due to (\ref{cpmcpn2}) $c_{p^m}(a)=c_{p^{m-1}}\left( a{\scriptstyle\mod p^{m-1}}\right)$ and then
\begin{equation}\label{tau^nondiv_p^m-1}
	\Bigl\langle \tau^{k} \mid f_{p^{m}}(\tau)\Bigr\rangle
	= \Bigl\langle \tau^{k} \mid \Phi_{p^{m}}(\tau)\cdot f_{p^{m-1}}(\tau)\Bigr\rangle
	\quad\text{ if } k\not\equiv0\md{p^{m-1}}.
\end{equation}

We'll consider odd $m$ and even separately.
\subsubsection*{For odd $m=2\mu+1$, $\mu\in\mathbb{N}$.} 
For example, if  we take  $k=0\in\mathbb{Z}(p^{m})$ then
\begin{equation*}
	\Bigl\langle \tau^{0} \mid f_{p^{m}}(\tau)\Bigr\rangle
	= p^\mu
	\qquad\text{ and }\qquad
	\Bigl\langle \tau^{0} \mid \Phi_{p^{m}}(\tau)\cdot f_{p^{m-1}}(\tau)\Bigr\rangle = p^\mu.
\end{equation*}
For  $k=p^{2\mu}\in\mathbb{Z}(p^{m})$
$$\begin{aligned}
	\Bigl\langle \tau^{p^{2\mu}} \mid f_{p^{m}}(\tau)\Bigr\rangle
	&= c_{p^{m}}(1\cdot p^{2\mu} ) = \left( 1+\chi(1) \right) p^\mu, \\
	\Bigl\langle \tau^{p^{2\mu}} \mid \Phi_{p^{m}}(\tau)\cdot f_{p^{m-1}}(\tau)\Bigr\rangle  &=
	\Bigl\langle \tau^{p^{2\mu}} \bigm| (1+{\tau^{p^{2\mu}}}+\cdots)\cdot (p^\mu+2\tau+\cdots) \Bigr\rangle = p^\mu,
\end{aligned}$$
so in general case when $k\equiv0\md{p^{m-1}}$
\begin{equation}\label{tau^div_p^m-1}
	\Bigl\langle \tau^{k} \mid f_{p^{m}}(\tau)\Bigr\rangle
	=  \Bigl\langle \tau^{k} \mid \Phi_{p^{m}}(\tau)\cdot f_{p^{m-1}}(\tau)\Bigr\rangle 
	+ \chi(\kappa)\, p^\mu\, ,
	\qquad\text{ where } k=\kappa\cdot p^{m-1}\in\mathbb{Z}(p^{m}).
\end{equation}
Relations 
(\ref{tau^nondiv_p^m-1}) and (\ref{tau^div_p^m-1}) give for odd $m$
\begin{equation}
	f_{p^m}(\tau) 
	= \Phi_{p^m}(\tau) \cdot f_{p^{m-1}}(\tau) +p^\mu \sum_{\kappa=0}^{p-1}\chi(\kappa) \tau^{\kappa\, p^{m-1}}
	= \Phi_{p^m}(\tau) \cdot f_{p^{m-1}}(\tau) + p^{\lfloor\frac{m}{2}\rfloor} \cdot g_p(\tau^{p^{m-1}}).
\end{equation}

\subsubsection*{For even $m=2\mu$, $\mu\in\mathbb{N}$.}
If  we take  $k=0\in\mathbb{Z}(p^{m})$ then due to (\ref{cpm})
\begin{equation*}
	\Bigl\langle \tau^{0} \mid f_{p^{m}}(\tau)\Bigr\rangle
	= p^\mu
	\qquad\text{ and }\qquad
	\Bigl\langle \tau^{0} \mid \Phi_{p^{m}}(\tau)\cdot f_{p^{m-1}}(\tau)\Bigr\rangle = p^{\mu-1}.
\end{equation*}
For  $k=p^{2\mu}\in\mathbb{Z}(p^{m})$
$$\begin{aligned}
	\Bigl\langle \tau^{p^{m-1}} \mid f_{p^{m}}(\tau)\Bigr\rangle
	&= 0, \\
	\Bigl\langle \tau^{p^{m-1}} \mid \Phi_{p^{m}}(\tau)\cdot f_{p^{m-1}}(\tau)\Bigr\rangle  &=
	\Bigl\langle \tau^{p^{2\mu-1}} \bigm| (1+{\tau^{p^{2\mu-1}}}+\cdots)\cdot (p^{\mu-1}+2\tau+\cdots) \Bigr\rangle = p^{\mu-1},
\end{aligned}$$
so in general case when $k\equiv0\md{p^{m-1}}$
\begin{equation}\label{tau^div_p^m-1+}
	\Bigl\langle \tau^{k} \mid f_{p^{m}}(\tau)\Bigr\rangle
	=  \Bigl\langle \tau^{k} \mid \Phi_{p^{m}}(\tau)\cdot f_{p^{m-1}}(\tau)\Bigr\rangle 
	+ p^\mu\cdot\delta(\kappa) - p^{\mu-1}\cdot1 ,
	\quad\text{ where } k=\kappa\cdot p^{m-1}\in\mathbb{Z}(p^{m}).
\end{equation}
We obtain
\begin{equation}
	\text{for } m=2\mu :\quad f_{p^m}(\tau) 
	= \Phi_{p^m}(\tau) \cdot f_{p^{m-1}}(\tau) + p^{\frac{m}{2}-1} \bigl(p -  \Phi_{p^m}(\tau)\bigr).
\end{equation}

\subsection{Multiplicities in $d=2$}
In two-dimensional case due to (\ref{pm_phi+g_^n}) from one-dimensional (\ref{fpm=Phipm_fpm-1+g^m}) we have
\begin{equation}
	\begin{aligned}
		f^2_{p^m} &= f_{p^m} \circ f_{p^m} =  \left(  \Phi_{p^m}\circ f_{p^{m-1}}\right) ^2 \,
		+ \Bigl(  g_p(\tau^{p^{m-1}}) \Bigr)^{2m}=\\
		&= p\cdot \Phi_{p^m}\cdot f^2_{p^{m-1}} \; + \left( (-1)^{\frac{p-1}2} \cdot\left(p -\Phi_{p^m}\right)\right) ^{m} =\\
		&=  p\cdot \Phi_{p^m}\cdot f_{p^{m-1}}^2 \; + (-1)^{\frac{p-1}{2}\,{m}} \cdot p^{m-1} \cdot\left(p -\Phi_{p^m}\right).
\end{aligned}\end{equation}

\subsection{Multiplicities in $d=3$}
In three-dimensional case
\begin{equation}\label{f3pm=Phi*f3pm-1}
	\begin{aligned}
		f^3_{p^m} = f_{p^m} \circ f_{p^m} \circ f_{p^m} 
		&= \left(  \Phi_{p^m}\circ f_{p^{m-1}}\right) ^3\, 
		+ (-1)^{\frac{p-1}{2}\,\lfloor\frac{m}{2}\rfloor\, m} \Bigl(  g_p(\tau^{p^{m-1}}) \Bigr)^{3m} =\\
		&= p^2\cdot \Phi_{p^m}\cdot f^3_{p^{m-1}} \, + (-1)^{\frac{p-1}{2}\lfloor\frac{m}{2}\rfloor m} \Bigl(  (-1)^{\frac{p-1}2}\, p\cdot g_p(\tau^{p^{m-1}}) \Bigr)^{m} =\\
		&= p^2\cdot \Phi_{p^m}\cdot f^3_{p^{m-1}} \, + (\pm)\, p^{m} 
		\Bigl( g_p(\tau^{p^{m-1}})\Bigr)^{m} ,\\
	\end{aligned}
\end{equation}
where the sign $\pm$ equals $(-1)$ if $\quad -p\equiv m\equiv1 \md 4\quad$ and $(+1)$ otherwise.

It is clear that if $f^3_{p^{m-1}}$ is divisible $p^{m}$ than $f^3_{p^{m}}$ is also divisible by $p^{m}$. In section \ref{N=p} the Theorem is proved for $N=p^1$ and hence by induction we proved (\ref{f3pm=Phi*f3pm-1}) the Theorem  for $N=p^m$.




\section{Definitions for proof-2}


\subsection{$d$-dimensional integer space}

Let us consider the $d$-dimensional space $\mathbb{Z}^d$ of integer vectors
$$
   \mbox{\boldmath$\xi$}=(\xi_1, \xi_2, ...,  \xi_d),\quad \xi_i \in \mathbb{Z}.
$$

In the figures, the elements of space $\mathbb{Z}^d$ is represented by cells of unit size in all coordinates.

Let $\mathrm L$ be some finite subset of $\mathbb{Z}^d$
$$
\mathrm L \subset \mathbb{Z}^d,\quad | \mathrm L | \in\mathbb{N},
$$
where $| \mathrm L |$ is the number of elements in $\mathrm L$.

Notation
$$
  n \cdot \mathrm L + \mathbf{a}=\{n\cdot\boldsymbol{\xi}+\mathbf{a}|\boldsymbol{\xi}\in\mathrm L\}.
$$
where $n \in \mathbb{N}$, $\mathbf{a} \in \mathbb{Z}^d$, is used for the set $\mathrm L$ stretched by factor $n$ and shifted by the vector $\mathbf{a}$.
That is, the points of this set are obtained from the points of the set $\mathrm L$ by multiplying all coordinates by $n$ and shifting by vectors $\mathbf{a}$:
$$
  \xi_i \rightarrow n\xi_i + a_i
$$

The notation $\mathrm L + \mathbf{a}$ is equivalent to $1 \cdot \mathrm L + \mathbf{a}$.

The notation $n \cdot \mathrm L$ is equivalent to $n \cdot \mathrm L + \mathbf{0}$, where $\mathbf{0}$ is the null vector in $\mathbb{Z}^d$.

\subsection{Sets of integers}

\begin{definition}
	$\mathbb{Z}(\alpha)$ is the set of integers from $0$ to $\alpha-1$, i.e.
	$$\mathbb{Z}(\alpha) = \left\{0, 1, 2, \ldots, \alpha-1 \right\}$$
\end{definition}
\begin{definition}
	$n\mathbb{Z}(\alpha) + a$ is the set of all numbers from $\mathbb{Z}(\alpha)$,
	multiplied by $k$ and shifted by $a$, i.e.
$$
   n\mathbb{Z}(\alpha)+a=\left\{a, n+a, 2n+a, \ldots, (\alpha-1)n+a\right\}.
$$
\end{definition}

The notation $\mathbb{Z}(\alpha)+a$ is equivalent to $1\mathbb{Z}(\alpha)+a$.

The notation $n\mathbb{Z}(\alpha)$ is equivalent to $n\mathbb{Z}(\alpha)+0$.

\subsection{Lattices}

Let us introduce the notation for the $b$-dimensional lattice 
$\underbrace{N \times N \times \ldots \times N}_b$.
\begin{definition}
	$\lattice^b(N)$ for $a\leqslant d$ is the set of $\boldsymbol{\xi}\in\mathbb{Z}^d$ such that $\xi_1,\dots,\xi_b\in\mathbb{Z}(N)$, and $\xi_{b+1}=\dots=\xi_d=0$.
\end{definition}

\begin{definition}
	$\lattice_{i_1\dots i_b}^b(N)$ for $b\leqslant d$ is the set of $\boldsymbol{\xi}\in\mathbb{Z}^d$ such that $\xi_{i_1},\dots,\xi_{i_b}\in\mathbb{Z}(N)$, and $\xi_j=0$ if $j\not\in\{i_1,\dots,i_b\}$.
\end{definition}

For example, if $d=6$, then $\lattice^3_{2,3,5}(N)$ is the set of $\boldsymbol{\xi}\in\mathbb{Z}^6$ such that $\xi_2,\xi_3,\xi_5\in\mathbb{Z}(\alpha)$, and $\xi_1=\xi_4=\xi_6=0$.

\subsection{Functions}



\begin{definition}
	$\ncells(\mathrm L, \alpha, k)$ ("the number of cells") is the number of cells 
	$\boldsymbol{\xi}\in\mathrm L\subset\mathbb{Z}^d$, 
	such that 
$$
  \sum_{i=1}^d \xi_i^2\equiv k \md\alpha.
$$
\end{definition}
In terms of $\ncells$ the multiplicity function $c_{Nd}$ \eqref{c(k)} has the following form
$$
  c_{Nd}(k)=\ncells(\lattice^d(N),N,k).
$$

\begin{definition}
	The logical function $\zrf(\mathrm L, \alpha, \beta)$ ("zerofy") is equal to ``$true$'' if 
	$$
  	 \forall k\in\mathbb{Z}(\alpha)\quad \ncells(\mathrm L, \alpha, k)\equiv 0 \md\beta,
	$$ 
	otherwise it is ``$false$''.
\end{definition}

The notation $\zrf(\mathrm L, \alpha, \beta)$ is equivalent to $\zrf(\mathrm L, \alpha, \beta) = true$.

The parameter $\alpha$ of the functions $\ncells$ and $\zrf$ is called the \textit{momentum ring parameter}, and $\beta$ is called the \textit{energy ring parameter}.

In terms of fumction $\zrf$ the theorem \eqref{threv} has the following form
$$
d\geqslant3\text{~~or~~}(d\geqslant2,N=2^n)~\Rightarrow~\zrf(\lattice^d(N),N,N).
$$

\section{Lemmas}

To simplify the proof we introduce several simple lemmas.
For obvious lemmas, the proofs are skipped.

\begin{lemma}[on increasing of the momentum ring parameter] \label{l1}
\be\label{lemma1-eq}
  \ncells(\mathrm L, \alpha, k) = \sum\limits_{j=0}^{n-1} \ncells(\mathrm L, n\alpha,  k + j\alpha).
\ee
\end{lemma}

\begin{proof}
	The number $k \in \mathbb{Z}(\alpha)$ corresponds to the numbers $k$, $k+\alpha$, $k+2\alpha$, $k+3\alpha$, ... , $k+( n-1)\alpha$ in $\mathbb{Z}(n\alpha)$. To go from modulus $n\alpha$ to modulus $\alpha$, one needs to sum the corresponding numbers of cells.
\end{proof}

\begin{lemma}[on decreasing of the momentum ring parameter] \label{l2}
$$	
\zrf(\mathrm L, n\alpha, \beta)~\Rightarrow~ \zrf(\mathrm L, \alpha, \beta).
$$
\end{lemma}


\begin{lemma}[on decreasing of the energy ring parameter] \label{l3}
$$
 \zrf(\mathrm L, \alpha, n\beta)~\Rightarrow~\zrf(\mathrm L, \alpha, \beta).
$$
\end{lemma}


\begin{lemma}[on lattice stretching] \label{l4}
$$
  \zrf(\mathrm L, \alpha, \beta)~\Rightarrow~{\zrf(n \cdot \mathrm L, n^2\alpha, \beta)}.
 $$
\end{lemma}


%
%
\begin{lemma}[on lattice shift] \label{l10}
$$
  \ncells(n\cdot\mathrm L+n\alpha\mathbf{a},n^2\alpha,k) = \ncells(n\cdot\mathrm L, n^2\alpha,k).
$$
\end{lemma}

\begin{proof}
	Let $\xi_i$ be the coordinates of a cell in the set $\mathrm L$. 
	$n\xi_i$ are the coordinates of a cell in the set $n\cdot \mathrm L$. 
	When shifted by the vector $n\alpha\mathbf{a}$, the coordinates of the cells of the set 
	$n\cdot\mathrm L$ are shifted by the corresponding components of this vector:
$$
  (n\xi_1,n\xi_2,n\xi_3,\ldots)\rightarrow(n\xi_1+n\alpha a_1,n\xi_2+n\alpha a_2,n\xi_3+n\alpha a_3,\ldots).
$$
	The sum of the squares of the coordinates changes as follows:
$$
  (n\xi_1)^2+(n\xi_2)^2+(n\xi_3)^2+\ldots\rightarrow
  (n\xi_1+n\alpha a_1)^2+(n\xi_2+n\alpha a_2)^2+(n\xi_3+n\alpha a_3)^2+\ldots.
$$
	Increment of the sum of squares:
\bean
  (n\xi_1+n\alpha a_1)^2+(n\xi_2+n\alpha a_2)^2+(n\xi_3+n\alpha a_3)^2+\ldots
  -(n\xi_1)^2-(n\xi_2)^2-(n\xi_3)^2-\ldots=\\
  =2n^2\alpha a_1\xi_1+n^2\alpha^2 a_1^2+2n^2\alpha a_2\xi_2+n^2\alpha^2 a_2^2+2n^2\alpha a_3\xi_3
  +n^2 \alpha^2 a_3^2+\ldots=\\
  =n^2\alpha(2a_1\xi_1+\alpha a_1^2+2a_2\xi_2+\alpha a_2^2+2a_3\xi_3+\alpha a_3^2+\ldots)
  \equiv 0 \md{n^2\alpha}.
\eean
	Therefore
$$
  \ncells(n\cdot \mathrm L+n\alpha\mathbf{a},n^2\alpha,k) = \ncells(n\cdot\mathrm L,n^2\alpha,k).
$$
\end{proof}

\begin{lemma}[on sum of lattices] \label{l5}
$$
  \mathrm L_1 \cap \mathrm L_2 = \varnothing,~\zrf(\mathrm L_1, \alpha, \beta), \zrf(\mathrm L_2, \alpha, \beta)~
  \Rightarrow~\zrf(\mathrm L_1 \cup \mathrm L_2, \alpha, \beta).
$$
\end{lemma}


\begin{lemma}[on subtraction of lattices] \label{l6}
$$
\mathrm L_1 \cap \mathrm L_2 = \varnothing,~\zrf(\mathrm L_1 \cup \mathrm L_2, \alpha, \beta), \zrf(\mathrm L_2, \alpha, \beta)~
\Rightarrow~\zrf(\mathrm L_1, \alpha, \beta).
$$
\end{lemma}


\begin{lemma}[on multiple lattices] \label{l7}~\\
	Let $\mathrm L_i \cap \mathrm L_j = \varnothing$ if $i \neq j$, and 
	$\ncells(\mathrm L_i,\alpha,k) = \ncells(\mathrm L_j,\alpha,k)$.
	Then 
$$	
  \zrf(\mathrm L_1, \alpha, \beta)~\Rightarrow~\zrf(\bigcup\limits_{j=1}^{n} \mathrm L_j, \alpha, n\beta).
$$
\end{lemma}


\begin{lemma}[on periodicity] \label{l8}~\\ 
Let 
$$
\forall k \in \mathbb{Z}(P\alpha)~~
{\ncells(\mathrm L, P\alpha, k \oplus P) 
= \ncells(\mathrm L, P\alpha, k)},
$$
where $\oplus$ is addition in the ring $\mathbb{Z}(P\alpha)$. Then 
$$ 
  \zrf(\mathrm L, P, \alpha \beta)~\Rightarrow~\zrf(\mathrm L, P\alpha, \beta).
$$
\end{lemma}

\begin{proof}
	By Lemma \ref{l1} (on increasing of the momentum ring parameter) and taking into account the periodicity of $\ncells(\mathrm L, P\alpha, k)$,
$$
  \ncells(\mathrm L, P,k) = \sum\limits_{j=0}^{\alpha-1} \ncells(\mathrm L, P\alpha, k + jP) 
  = \alpha \cdot \ncells(\mathrm L, P\alpha, k),
$$
$$
  \frac{\ncells(\mathrm L, P, k)}{\alpha} = \ncells(\mathrm L, P\alpha, k).
$$
	If $\ncells(\mathrm L, P, k)$ is a multiple of $\alpha\beta$, 
	then $\ncells(\mathrm L, P\alpha, k)$ is a multiple of $\beta$.
\end{proof}

\begin{lemma}[on permutation of remainders] \label{l11}
	Let $\alpha$ and $n$ be coprime, $A$ be the set the remainders of the division of elements $n\mathbb{Z}(\alpha)+a$ by $\alpha$. Then $A = \mathbb{Z}(\alpha)$.
\end{lemma}


\section{Proof-2}

Theorem on number-theory renormalization (\ref{threv}) formulated in terms of the function $\zrf$: for $d\geqslant3$ and any natural $N$, as well as for $d=2$ and $N=2^ m, m \in \mathbb{N}$
\begin{equation}\label{eq4}
   \zrf(\lattice^d(N), N, N).
\end{equation}

The theorem is obvious in the trivial case $N=1$, easily verified for $N=2$, and proved in the Appendix \ref{N=p} for odd primes $N$. It remains to prove the theorem for composite $N$.

We will prove the theorem separately for each of the following $3$ cases of composite $N$:
\begin{itemize}
	\item
	$N$ is a power of an odd prime $p$,
	\item
	$N$ is a power of $2$,
	\item
	$N$ is not a prime power.
\end{itemize}

\subsection{Powers of odd primes}

Let us prove the theorem (\ref{eq4}) for the $m$-th power of an odd prime $p$:
\begin{equation}\label{eq1}
  N = p^m.
\end{equation}

Unless otherwise specifically stated, we will assume that $d=3$.

We will prove by induction, assuming that the hypothesis is true for $p^{m-1}$ и для $p^{m-2}$:
\begin{equation}\label{eq5}
  \zrf(\lattice^d(p^{m-1}), p^{m-1}, p^{m-1}),
\end{equation}
\begin{equation}\label{eq6}
  \zrf(\lattice^d(p^{m-2}), p^{m-2}, p^{m-2}).
\end{equation}
\par
The basis of the induction is the proven above validity of the hypothesis for $d=3$ for ${m=0,1}$, i.e. for $N=1$ and $N=p$:
\begin{equation}\label{eq8}
  \zrf(\lattice^d(1), 1, 1),
\end{equation}
\begin{equation}\label{eq7}
  \zrf(\lattice^d(p), p, p).
\end{equation}

\begin{figure}[H]
	\centering
	\includegraphics[width=1\linewidth]{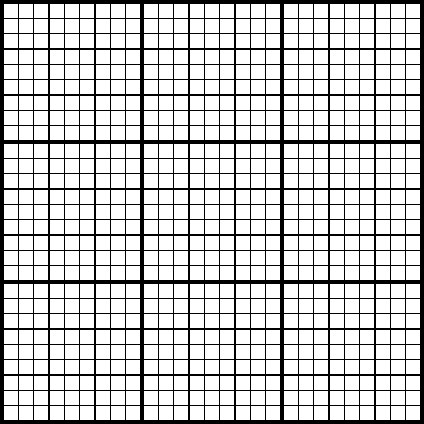}
	\caption{2-dimensional $\lattice^2(N)$ (the section of $\lattice^3(N)$). The origin is in the lower left corner. Here is an example $N=3^3=27$.}
	\label{fig:p1}
\end{figure}

Divide the $\lattice^d(N)$ (Fig.~\ref{fig:p1}) into 2 parts: $p \cdot \lattice^d(p^{m-1})$ and $\lattice^d(N ) - p \cdot \lattice^d(p^{m-1})$.

The first part ($p \cdot \lattice^d(p^{m-1})$) consists of all cells of $\lattice^d(N)$ such that each of their coordinates $\xi_i$ is a multiple of $p$ (Fig.~\ref{fig:p2}).

The second part ($\lattice^d(N) - p \cdot \lattice^d(p^{m-1})$) consists of all other cells of $\lattice^d(N)$, i.e. those with at least one of their coordinates $\xi_i$ is not a multiple of $p$ (Fig.~\ref{fig:p3}).

Similarly, we divide the $\lattice^d(p^{m-1})$ into 2 parts: $p \cdot \lattice^d(p^{m-2})$ and $\lattice^d(p^{m-1 }) - p \cdot \lattice^d(p^{m-2})$. (Fig.~\ref{fig:p4})

\begin{figure}[H]
	\centering
	\includegraphics[width=1\linewidth]{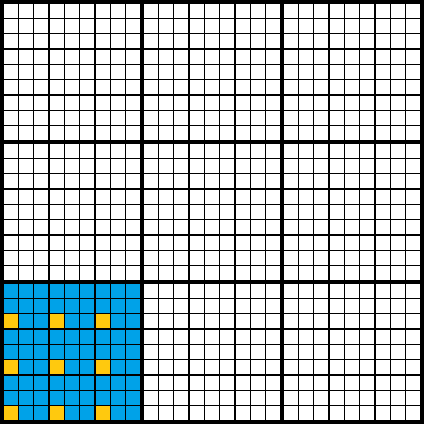}
	\caption{$\lattice^2(p^{m-1})$ (orange and blue cells) divided into parts $p \cdot \lattice^2(p^{m-2})$ (orange cells) and ${\lattice ^2(p^{m-1}) - p \cdot \lattice^2(p^{m-2})}$ (blue cells).}
	\label{fig:p4}
\end{figure}

\par
We will prove zerofying for each of the 2 parts of $\lattice^d(N)$ separately:
\begin{equation}\label{eq11}
  \zrf(p \cdot \lattice^d(p^{m-1}), N, N),
\end{equation}
\begin{equation}\label{eq12}
  \zrf(\lattice^d(N) - p \cdot \lattice^d(p^{m-1}), N, N),
\end{equation}
and then, by Lemma \ref{l5} (on sum of lattices), zerofying will be proved for the entire $\lattice^d(N)$ (\ref{eq4}).

\subsubsection*{The first part of the lattice}

\begin{figure}[H]
	\centering
	\includegraphics[width=1\linewidth]{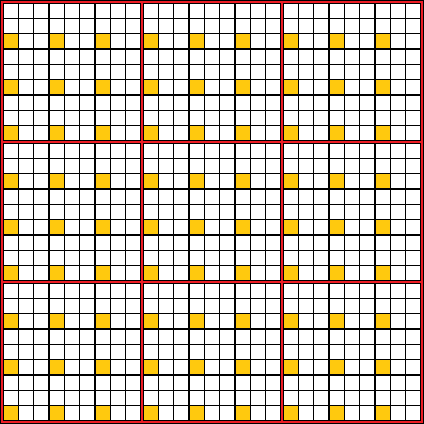}
	\caption{The first part ($p \cdot \lattice^2(p^{m-1})$) of $\lattice^2(N)$ (orange cells). The red frames show the partition of $p \cdot \lattice^2(p^{m-1})$ into sets $p \cdot \lattice^2(p^{m-2}) + p^{m-1}\mathbf{ a}$, where $\mathbf{a} \in \lattice^2(p)$.}
	\label{fig:p2}
\end{figure}

Since the set $p \cdot \lattice^d(p^{m-2})$ is $\lattice^d(p^{m-2})$ stretched by $p$ times, by virtue of the hypothesis for $ p^{m-2}$ (\ref{eq6}), Lemma \ref{l4} (on lattice stretching) and the fact that $N=p^m$ (\ref{eq1})
\begin{equation}\label{eq13}
  \zrf(p \cdot \lattice^d(p^{m-2}), N, p^{m-2}).
\end{equation}

Note that the set $p \cdot \lattice^d(p^{m-1})$ can be divided into $p^d$ subsets $p \cdot \lattice^d(p^{m-2}) + p^{ m-1}\mathbf{a}$, where $\mathbf{a} \in \lattice^d(p)$. One of these subsets is $p \cdot \lattice^d(p^{m-2})$, and the rest are obtained from $p \cdot \lattice^d(p^{m-2})$ by such a shift of all its cells, that each coordinate difference is a multiple of $p^{m-1}$ (Fig.~\ref{fig:p2}).

Since $p^{m-1}$ and the coordinates of all cells from $p \cdot \lattice^d(p^{m-2})$ are divisible by $p$, by Lemma \ref{l10} (on lattice shift) for each set ${p \cdot \lattice^d(p^{m-2}) + p^{m-1}\mathbf{a}}$ the value of $\ncells(p \cdot \lattice^d(p^ {m-2}) + p^{m-1}\mathbf{a}, N, k)$ does not depend on $\mathbf{a}$. Therefore, since one of the sets ${p \cdot \lattice^d(p^{m-2}) + p^{m-1}\mathbf{a}}$ is the set $p \cdot \lattice^d(p ^{m-2})$, and there are $p^d$ such sets, due to (\ref{eq13}) and Lemma \ref{l7} (on multiple lattices)
\begin{equation}\label{eq19}
  \zrf(p \cdot \lattice^d(p^{m-1}), N, p^{m+d-2}).
\end{equation}
Hence, taking into account that $N=p^m$ (\ref{eq1}), if $d \geqslant 2$, by Lemma \ref{l3} (on decreasing of the energy ring parameter)
\begin{equation}\label{eq20}
   \zrf(p \cdot \lattice^d(p^{m-1}), N, N).
\end{equation}

\subsubsection*{The second part of the lattice}

\begin{figure}[H]
	\centering
	\includegraphics[width=1\linewidth]{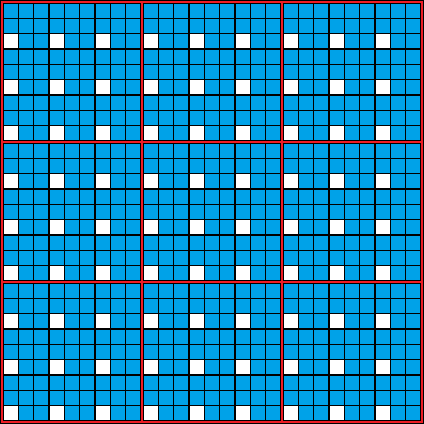}
	\caption{The second part ($\lattice^2(N) - p \cdot \lattice^2(p^{m-1})$) of the $\lattice^2(N)$ (blue cells). The red frames show the partition of ${\lattice^2(N) - p \cdot \lattice^2(p^{m-1})}$ into sets $[{\lattice^2(p^{m-1}) - p \cdot \lattice^2(p^{m-2})] + p^{m-1}\mathbf{a}}$, where $\mathbf{a} \in \lattice^2(p)$.}
	\label{fig:p3}
\end{figure}

From (\ref{eq13}) and by Lemma \ref{l2} (on decreasing of the momentum ring parameter) it follows that
\begin{equation}\label{eq21}
   \zrf(p \cdot \lattice^d(p^{m-2}), p^{m-1}, p^{m-2}).
\end{equation}

From the validity of the hypothesis for $p^{m-1}$ (\ref{eq5}) and by virtue of the Lemma \ref{l3} (on decreasing of the energy ring parameter), it follows that
\begin{equation}\label{eq22}
  \zrf(\lattice^d(p^{m-1}), p^{m-1}, p^{m-2}).
\end{equation}

Due to (\ref{eq21}), (\ref{eq22}) and Lemma \ref{l6} (on subtraction of lattices)
\begin{equation}\label{eq25}
  \zrf(\lattice^d(p^{m-1}) - p \cdot \lattice^d(p^{m-2}), p^{m-1}, p^{m-2}).
\end{equation}

Now we divide the set $\lattice^d(N) - p \cdot \lattice^d(p^{m-1})$ into subsets $[\lattice^d(p^{m-1}) - p \cdot \lattice^d(p^{m-2})] + p^{m-1}\mathbf{a}$, where $\mathbf {a} \in \lattice^d(p)$ (Fig.~\ref{fig:p3}) in the same way as the set $p \cdot \lattice^d(p^{m-1})$ was divided into sets $p \cdot \lattice^d(p^{m-2}) + p^{m-1}\mathbf{a}$, where $\mathbf{a} \in \lattice^d(p)$.

Similarly, by Lemma \ref{l10} (on lattice shift), for each of these sets the values of $\ncells([\lattice^d(p^{m-1}) - p \cdot \lattice^d(p^{m- 2})] + p^{m-1}\mathbf{a}, p^{m-1}, k)$ does not depend on $\mathbf{a}$. 
Therefore, since one of these sets is ${\lattice^d(p^{m-1}) - p \cdot \lattice^d(p^{m-2})}$, and there are $p^d $ such sets, due to (\ref{eq25}) and Lemma \ref{l7} (on multiple lattices)
\begin{equation}\label{eq26}
  \zrf(\lattice^d(N) - p \cdot \lattice^d(p^{m-1}), p^{m-1}, p^{m+d-2}).
\end{equation}
For $d \geqslant 3$, if we succeed in applying Lemma \ref{l8} (on periodicity), then, taking into account that $N=p^m$ (\ref{eq1}), zeroing for the second part of the lattice (\ref {eq12}) will be proven. For $d = 2$ a separate consideration is required.

Now we divide the set $\lattice^d(N) - p \cdot \lattice^d(p^{m-1})$ into disjoint sets of the form
\bean
  p^{m-1} \cdot lattice(p) + (\nu_1, \xi_2, \xi_3, \ldots),\\
  p^{m-1} \cdot lattice_2(p) + (\xi_1, \nu_2, \xi_3, \ldots),\\
  p^{m-1} \cdot lattice_3(p) + (\xi_1, \xi_2, \nu_3, \ldots),\\
  \ldots\qquad\qquad\qquad
\eean
$$
  \nu_i \in \mathbb{Z}({p^{m-1}}),\quad \nu_i \mod{p} \neq 0.
$$
with $p$ cells in each set (Fig.~\ref{fig:p5}) as follows. Let's take an arbitrary cell $\mbox{\boldmath$\xi$}$ from $\lattice^d(N) - p \cdot \lattice^d(p^{m-1})$, take from its $d$ coordinates the first not multiple of $p$ (according to the definition of this set, at least one of the coordinates is not a multiple of $p$), and then we will shift this coordinate by multiples of $p^{m-1}$. Collect the resulting $p$ cells into a set. Regardless of which point of this set we start from, we will collect the same set. So, the partition of $\lattice^d(N) - p \cdot \lattice^d(p^{m-1})$ into subsets is unique.

Suppose we have selected a cell $\boldsymbol\xi$, and $\xi_1$ is the first coordinate not a multiple of $p$. We will shift $\xi_1$ by multiples of $p^{m-1}$. We get the set 
$p^{m-1} \cdot \lattice(p) + (\nu_1, \xi_2, \xi_3, \ldots)$. The elements of the set has the form
\begin{equation}\label{eq27}
   (\xi_1, \xi_2, \xi_3, \ldots)= (\nu_1 + kp^{m-1}, \xi_2, \xi_3, \ldots),
\end{equation}
where $k \in \mathbb{Z}(p)$, $\nu_1 \in \mathbb{Z}({p^{m-1}})$, $\nu_i \mod{p}\ne0$.

The sum of squares of the coordinates (\ref{eq27}) of cells from this set is
\be\label{eq28}
   (\nu_1 + kp^{m-1})^2 + \xi_2^2 + \xi_3^2 + \ldots=
   k^2p^{2m-2} + 2kp^{m-1}\nu_1 + \nu_1^2 + \xi_2^2 + \xi_3^2 + \ldots
\ee

\begin{figure}[H]
	\centering
	\includegraphics[width=1\linewidth]{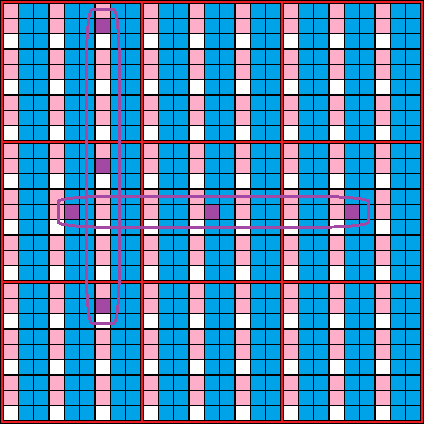}
	\caption{Cells of the the set $\lattice^2(N) - p \cdot \lattice^2(p^{m-1})$ whose first coordinate is not a multiple of $p$ are shown in blue. The pink color shows those cells of this set, for which the first coordinate is a multiple of $p$, and the second coordinate is not a multiple of $p$. (Note: a 3-dimensional set also has cells whose first two coordinates are not a multiple of $p$, and the third coordinate is a multiple of $p$). Purple color shows examples of sets $p^{m-1} \cdot \lattice(p) + (\nu_1, \xi_2)$ (arranged horizontally) and $p^{m-1} \cdot \lattice_2(p) + (\xi_1, \nu_2)$ (arranged vertically).}
	\label{fig:p5}
\end{figure}

Consider the sum \eqref{eq28} modulo $N$. If $m \geqslant 2$, then $p^{2m-2}$ is a multiple of $N=p^m$, 
so the term $k^2p^{2m-2}$ can be dropped:
\begin{equation}\label{eq30}
   (2kp^{m-1}\nu_1 + \nu_1^2 + \xi_2^2 + \xi_3^2 + \ldots) \mod{N}.
\end{equation}

Consider the the remainder of the division of the term $2kp^{m-1}\nu_1$ by $N$ (the other terms of (\ref{eq30}) are constant):
\begin{equation}\label{eq31}
   2kp^{m-1}\nu_1 \mod{N} = p^{m-1}(2k\nu_1 \mod{p}).
\end{equation}
Since $\nu_i \mod{p} \neq 0$ and $p$ are odd primes, $2\nu_1$ and $p$ are coprime. Therefore, by Lemma \ref{l11} (on permutation of remainders), the set of possible values of $2k\nu_1 \mod{p}$ coincides with $\mathbb{Z}(p)$. Therefore, the set of possible values (\ref{eq31}) is the same as $p^{m-1}\mathbb{Z}(p)$.

Therefore, for the set $p^{m-1} \cdot \lattice(p) + (\nu_1, \xi_2, \xi_3, \ldots)$ the function $\ncells(p^{m-1} \cdot \lattice(p) + (\nu_1, \xi_2, \xi_3, \ldots), N,k)$ is periodic in $k$ with period $p^{m-1}$:
\begin{multline}\label{eq38}
   \ncells(p^{m-1} \cdot \lattice(p) + (\nu_1, \xi_2, \xi_3, \ldots), N, k\oplus p^{m-1}) = \\ = \ncells(p^{m-1} \cdot \lattice(p) + (\nu_1, \xi_2, \xi_3, \ldots), N, k).
\end{multline}
The function \eqref{eq38} takes the values $0$ and $1$, and it takes the values $1$ only for $k-(\nu_1^2+\xi_2^2+\xi_3^2+\dots)\in p^{m-1}\mathbb{Z}(p)$.

Since the set $\lattice^d(N) - p \cdot \lattice^d(p^{m-1})$ consists entirely of such disjoint sets
with property \eqref{eq38}, 
and for union of disjoint sets, these functions add up, then for the entire set ${\lattice^d(N) - p \cdot \lattice^d(p^{m-1})}$ this function is periodic in $k$ with period $p ^{m-1}$, that is
\begin{multline}\label{eq39}
  \ncells(\lattice^d(N) - p \cdot \lattice^d(p^{m-1}), N, k \oplus p^{m-1}) = \\ 
  = \ncells(\lattice^d(N) - p \cdot \lattice^d(p^{m-1}), N, k).
\end{multline}
Hence, by virtue of (\ref{eq26}) and Lemma \ref{l8} (on periodicity)
\begin{equation}\label{eq40}
  \zrf(\lattice^d(N) - p \cdot \lattice^d(p^{m-1}), N, p^{m+d-3}).
\end{equation}
If $d \geqslant 3$, then, taking into account that $N=p^m$ (\ref{eq1}), by virtue of Lemma \ref{l3} (on decreasing of the energy ring parameter) it follows from (\ref{eq40}) that
\begin{equation}\label{eq41}
   \zrf(\lattice^d(N) - p \cdot \lattice^d(p^{m-1}), N, N).
\end{equation}
Thus, applying Lemma \ref{l5} (on sum of lattices) to \eqref{eq20} and \eqref{eq41} we have proved the theorem for $N=p^m$, $d\geqslant3$
$$
  \zrf(\lattice^d(N), N, N).
$$

\subsection{Proof for $d\geqslant2$, $N=2^m$}

If we apply the proof for powers of primes for the case $d=2$, $p=2$, then it turns out that the proof for the first part of the lattice works, but the proof for the second part of the lattice is valid up to the formula (\ref{eq26}) before the assumptions $p\not=2$ and $d\geqslant3$ are used. A closer look shows that for $d=p=2$ we need to retreat to the formula \eqref{eq19}.

Now, as an induction hypothesis, we take the validity of the hypothesis not for the previous two powers ($m-1$, $m-2$), but for three ($m-1$, $m-2$, $m-3$):
\begin{equation}\label{eq76}
	\zrf(\lattice^d(p^{m-1}), p^{m-1}, p^{m-1}),
\end{equation}
\begin{equation}\label{eq77}
	\zrf(\lattice^d(p^{m-2}), p^{m-2}, p^{m-2}),
\end{equation}
\begin{equation}\label{eq78}
	\zrf(\lattice^d(p^{m-3}), p^{m-3}, p^{m-3}).
\end{equation}
Here $d=p=2$, but we do not substitute $2$ for these parameters for the convenience of comparison with the previously analyzed cases.

Since we now use the hypothesis for the previous three powers as the induction hypothesis, we can use the proof using the previous two powers by substituting $p^{m-1}$ instead of $N$ (that is, lowering the required powers $p$ by $1$). For $d\geqslant2$ from (\ref{eq19}) by reducing the powers of $p$ by $1$ it follows
\begin{equation}\label{eq79}
	\zrf(p \cdot \lattice^d(p^{m-2}), p^{m-1}, p^{m-1}).
\end{equation}

From the validity of the hypothesis for $p^{m-1}$ (\ref{eq76}), due to (\ref{eq79}) and the Lemma \ref{l6} (on subtraction of lattices)
\begin{equation}\label{eq80}
	\zrf(\lattice^d(p^{m-1}) - p \cdot \lattice^d(p^{m-2}), p^{m-1}, p^{m-1}).
\end{equation}

In the same way as (\ref{eq26}) was obtained from (\ref{eq25}), we obtain \eqref{eq81} from (\ref{eq76}).

Divide the set $\lattice^d(N) - p \cdot \lattice^d(p^{m-1})$ into subsets
$[\lattice^d(p^{m-1}) - p \cdot \lattice^d(p^{m-2})] + p^{m-1}\mathbf{a}$, where $\mathbf{a} \in \lattice^d(p)$ (Fig.~\ref{fig:p3}).

Similarly, by Lemma \ref{l10} (on lattice shift), for each of these sets the values of $\ncells([\lattice^d(p^{m-1}) - p \cdot \lattice^d(p^{m- 2})] + p^{m-1}\mathbf{a}, p^{m-1}, k)$ does not depend on $\mathbf{a}$. Therefore, since one of these sets is ${\lattice^d(p^{m-1}) - p \cdot \lattice^d(p^{m-2})}$, and there are $p^d$ such sets, due to (\ref{eq80}) and Lemma \ref{l7} (on multiple lattices)

\begin{equation}\label{eq81}
	\zrf(\lattice^d(N) - p \cdot \lattice^d(p^{m-1}), p^{m-1}, p^{m+d-1}).
\end{equation}

Now (similarly to derivation of \eqref{eq26}) we need to apply the \ref{l8} Lemma (on periodicity) to increase the momentum ring parameter by decrease the energy ring parameter.

For $p=2$, the set of possible values (\ref{eq31}) no longer coincides with $p^{m-1}\mathbb{Z}(p)$, which was based on the fact that $2\nu_1$ and $p$ are coprime, but this is not true for the case $p=2$.

Therefore, for $p=2$ we will divide into sets of $p$ cells not the set $\lattice^d(N) - p \cdot \lattice^d(p^{m-1})$ as a whole, but each of the $p^d$ subsets
$[\lattice^d(p^{m-1})-p\cdot \lattice^d(p^{m-2})]+p^{m-1}\mathbf{a}$,
where $\mathbf{a} \in \lattice^d(p)$.
So, the coordinate in each of the resulting sets will be shifted by a value that is a multiple of $p^{m-2}$, not $p^{m-1}$.

The set $\lattice^d(N) - p \cdot \lattice^d(p^{m-1})$ is divided into disjoint sets of the form
$$
p^{m-2} \cdot \lattice(p) + (\nu_1, \xi_2, \xi_3, \ldots),
$$
$$
p^{m-2} \cdot \lattice_2(p) + (\xi_1, \nu_2, \xi_3, \ldots),
$$
$$
p^{m-2} \cdot \lattice_3(p) + (\xi_1, \xi_2, \nu_3, \ldots),
$$
$$
\ldots
$$
$$
\nu_i \in \mathbb{Z}({p^{m-2}})+p^{m-1}\mathbb{Z}(p),\qquad
\nu_i \mod{p}\ne0.
$$
with $p=2$ cells in each.

This partition is performed as follows. Take an arbitrary cell $\boldsymbol{\xi}$ from $\lattice^d(N) - p \cdot \lattice^d(p^{m-1})$,
select from its $d$ coordinates the first non-multiple of $p$ (according to the definition of this set, at least one of the coordinates is not a multiple of $p$), and then we will shift this coordinate by multiples of $p^{m-2}$, so as not to go beyond a specific set 
$[\lattice^d(p^{m-1}) - p \cdot \lattice^d(p^{m-2})] + p^{m-1}\mathbf{a}$ (Fig.~\ref{fig:p6}). 
Collect the resulting $p$ cells into a set. Regardless of which point of this set we start from, we will collect the same set, that is, such a partition of $\lattice^d(N) - p \cdot \lattice^d(p^{m-1})$ into subsets is unique.

\begin{figure}[H]
	\centering
	\includegraphics[width=1\linewidth]{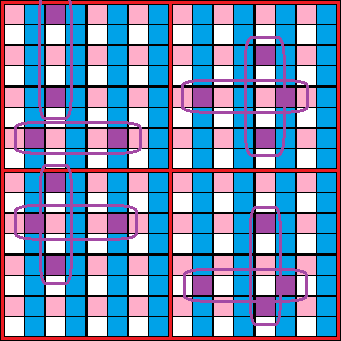}
	\caption{Cells of the $\lattice^2(N) - p \cdot \lattice^2(p^{m-1})$ with first coordinate is not a multiple of $p$ are shown in blue. The pink color shows those cells of this set, for which the first coordinate is a multiple of $p$, and the second coordinate is not a multiple of $p$. (Note: a 3-dimensional set also has cells whose first two coordinates are not a multiple of $p$, and the third coordinate is a multiple of $p$). The set $\lattice^2(N) - p \cdot \lattice^2(p^{m-1})$ is divided into sets $[\lattice^2(p^{m-1}) - p \cdot \lattice^2( p^{m-2})] + p^{m-1}\mathbf{a}$, where $\mathbf{a} \in \lattice^2(p)$ (shown using red frames), each of of which is divided into sets $p^{m-2} \cdot \lattice(p) + (\nu_1, \xi_2)$ (arranged horizontally) and ${p^{m-2} \cdot \lattice_2(p) + (\xi_1, \nu_2)}$ (arranged vertically) examples of which are shown in purple. Here $N=2^4=16$.}
	\label{fig:p6}
\end{figure}

Suppose we have selected a cell $\boldsymbol\xi$, and the first coordinate not a multiple of $p$ is $\xi_1$. We will shift $\xi_1$ by multiples of $p^{m-2}$. We get the set $p^{m-2} \cdot \lattice(p) + (\nu_1, \xi_2, \xi_3, \ldots)$ from the cells
\begin{equation}\label{eq42}
	(\xi_1, \xi_2, \xi_3, \ldots)=(\nu_1 + kp^{m-2}, \xi_2, \xi_3, \ldots),
\end{equation}
where $k \in \mathbb{Z}(p)$,
$\nu_1 \in \mathbb{Z}({p^{m-2}})+p^{m-1}\mathbb{Z}(p),$
$\nu_1 \mod{p}\ne0$. 

The sum of squares of the coordinates (\ref{eq42}) of cells from this set is
\begin{equation}
	(\nu_1 + kp^{m-2})^2 + \xi_2^2 + \xi_3^2 + \ldots
\label{eq44}
	=k^2p^{2m-4} + 2kp^{m-2}\nu_1 + \nu_1^2 + \xi_2^2 + \xi_3^2 + \ldots
\end{equation}

Consider this sum of squares modulo $N$. If $m \geqslant 4$, then $p^{2m-4}$ is a multiple of $N$, and the term $k^2p^{2m-4}$ can be dropped:
\begin{equation}\label{eq45}
	(2kp^{m-2}\nu_1 + \nu_1^2 + \xi_2^2 + \xi_3^2 + \ldots) \mod{N}
\end{equation}

Due to the fact that it is possible to pass from the expression (\ref{eq44}) to the expression (\ref{eq45}) only when $m \geqslant 4$, for $p=2$ we have to choose as induction basis not $ m=0,1,2$, but $m=1,2,3$.

Taking into account that $p=2$, so $\nu_1=1$ (since $\nu_1 \mod{p}\ne0$), we rewrite (\ref{eq45}) as
\begin{equation}\label{eq46}
	(kp^{m-1} + 1 + \xi_2^2 + \xi_3^2 + \ldots) \mod{N}.
\end{equation}

Consider the remainder of the division of the term $2kp^{m-1}\nu_1$ by $N$ (the other terms of (\ref{eq46}) are constant):
\begin{equation}\label{eq47}
	kp^{m-1}\in p^{m-1}\mathbb{Z}(p).
\end{equation}

Therefore, for the set $p^{m-2} \cdot \lattice(p) + (\nu_1, \xi_2, \xi_3, \ldots)$ the function $\ncells(p^{m-2} \cdot \lattice(p) + (\nu_1, \xi_2, \xi_3, \ldots), N, k)$ is periodic in $k$ with period $p^{m-1}$:
\begin{multline}\label{eq48}
	\ncells(p^{m-2} \cdot \lattice(p) + (\nu_1, \xi_2, \xi_3, \ldots),N, k \oplus p^{m-1}) = \\ = \ncells(p^{m-2} \cdot \lattice(p) + (\nu_1, \xi_2, \xi_3, \ldots), N, k).
\end{multline}

Since the set $\lattice^d(N) - p \cdot \lattice^d(p^{m-1})$ consists entirely of such disjoint sets, for each of which this function is periodic with period $p^{m-1 }$, and for union of disjoint sets, these functions add up, then for the entire set ${\lattice^d(N) - p \cdot \lattice^d(p^{m-1})}$ this function is periodic in $k$ with period $p ^{m-1}$, that is
\begin{multline}\label{eq49}
	\ncells(\lattice^d(N) - p \cdot \lattice^d(p^{m-1}),N,k\oplus p^{m-1}) = \\ 
   =\ncells(\lattice^d(N) - p \cdot \lattice^d(p^{m-1}),N,k).
\end{multline}

Further, similarly to the statement \eqref{eq40} 
we obtain
\begin{equation}\label{eq40-bin}
	\zrf(\lattice^d(N) - p \cdot \lattice^d(p^{m-1}), N, N).
\end{equation}

The induction base ($N=2,4,8$) is verified numerically.

\subsection{Values $\ncells(\lattice^d(N),N,k)$ for $d=1,2$, $N=2,4,8$}
\begin{center}
	\begin{tabular}{|r|c|c|}
		\hline
		$k$ & $0$ & $1$ \\
		\hline
		$\ncells(\lattice^1(2),2,k)$ & $1$ & $1$ \\
		\hline
		$\ncells(\lattice^2(2),2,k)$ & $2$ & $2$ \\
		\hline
	\end{tabular}
\end{center}

\begin{center}
	\begin{tabular}{|r|c|c|c|c|}
		\hline
		$k$ & $0$ & $1$ & $2$ & $3$ \\
		\hline
		$\ncells(\lattice^1(4),4,k)$ & $2$ & $2$ & $0$ & $0$ \\
		\hline
		$\ncells(\lattice^2(4),4,k)$ & $4$ & $8$ & $4$ & $0$ \\
		\hline
	\end{tabular}
\end{center}

\begin{center}
	\begin{tabular}{|r|c|c|c|c|c|c|c|c|}
		\hline
		$k$ & $0$ & $1$ & $2$ & $3$ & $4$ & $5$ & $6$ & $7$ \\
		\hline
		$\ncells(\lattice^1(8),8,k)$ & $2$ & $4$ & $0$ & $0$ & $2$ & $0$ & $0$ & $0$ \\
		\hline
		$\ncells(\lattice^2(8),8,k)$ & $8$ & $16$ & $16$ & $0$ & $8$ & $16$ & $0$ & $0$ \\
		\hline
	\end{tabular}
\end{center}

\subsection{Composite numbers that are not prime powers}

A composite number $N$ can be represented as a product of powers of unequal primes $p_i$ (if $i \neq j$, then $p_i \neq p_j$):
\begin{equation}\label{eq53}
   N=\prod\limits_{i=1}^n p_i^{m_i}.
\end{equation}
Denote
\begin{equation}\label{eq54}
   M=\prod\limits_{i=1}^{n-1} p_i^{m_i},
\end{equation}
\begin{equation}\label{eq55}
   K=p_n^{m_n}.
\end{equation}
That is
\begin{equation}\label{eq56}
   N=MK.
\end{equation}
$M$ and $K$ are coprime.

We will prove by induction. As a step of induction, we deduce the validity of the hypothesis for $N$ from its validity for $M$ and $K$:
\begin{equation}\label{eq57}
  \zrf(\lattice^d(M), M, M),~\zrf(\lattice^d(K), K, K)~~\Rightarrow~~
  \zrf(\lattice^d(N), N, N).
\end{equation}
As the basis of the induction, we take the validity of the hypothesis for $p_1^{m_1}$:
\begin{equation}\label{eq58}
  \zrf(\lattice^d(p_1^{m_1}), p_1^{m_1}, p_1^{m_1}).
\end{equation}
Since $p_1^{m_1}$ and $K$ are powers of primes, for $d=3$ the base of induction has already been proved and $\zrf(\lattice^d(K), K, K)$ is proved.

Represent the lattice $\lattice^d(N)$ as $K^d$ pairwise disjoint lattices $K \cdot \lattice^d(M) + \mathbf{a}$, where $\mathbf{a} \in \lattice^d (K)$ (Fig.~\ref{fig:p7}).

We prove that for each lattice $K \cdot \lattice^d(M) + \mathbf{a}$
\begin{equation}\label{eq60}
  \zrf(K \cdot \lattice^d(M) + \mathbf{a},N,N)
\end{equation}
after which, by Lemma \ref{l5} (on sum of lattices), it will be proved that
\begin{equation}\label{eq61}
  \zrf(\lattice^d(N),N,N).
\end{equation}

Each lattice coordinate $\lattice^d(M)$ runs through all values in $\mathbb{Z}(M)$. By Lemma \ref{l11} (on permutation of remainders), each coordinate of the lattice $K \cdot \lattice^d(M) + \mathbf{a}$ modulo $M$ also runs through all values in $\mathbb{ Z}(M)$. That is, the lattices $\lattice^d(M)$ and $K \cdot \lattice^d(M) + \mathbf{a}$ consist of the same cells if their coordinates are taken modulo $M$. Therefore, the validity of the hypothesis for $M$ (\ref{eq57}) implies
\begin{equation}\label{eq69}
  \zrf(K \cdot \lattice^d(M) + \mathbf{a}, M, M)
\end{equation}

Consider two cells in $K \cdot \lattice^d(M) + \mathbf{a}$:
$$
  K\boldsymbol\xi + \mathbf{a},
$$
$$
  K\boldsymbol{\xi^\prime} + \mathbf{a},
$$
where $\boldsymbol\xi, \boldsymbol{\xi^\prime} \in \lattice^d(M)$. The difference of the sums of squares of their coordinates:
\bea\label{eq70}
  (K\xi_1+a_1)^2 + (K\xi_2+a_2)^2 + \ldots - (K\xi_1^\prime+a_1)^2 - (K\xi_2^\prime+a_2)^2 - \ldots=\\
\nonumber
  =K^2(\xi_1^2 + \xi_2^2 + \ldots - {\xi_1^\prime}^2 - {\xi_2^\prime}^2 - \ldots) + 2K((\xi_1-\xi_1^\prime)a_1 + (\xi_2-\xi_2^\prime)a_2 + \ldots).
\eea

So, the difference between the sums of squared coordinates of two cells from ${K \cdot \lattice^d(M) + \mathbf{a}}$ is a multiple of $K$.

Compare the functions $\ncells(K \cdot \lattice^d(M) + \mathbf{a}, M, z)$ and 
$\ncells(K \cdot \lattice^d(M) + \mathbf{a}, N, y)$.

Each $z\in\mathbb{Z}(M)$ corresponds to $K$ different values of $y\in\mathbb{Z}(N)$ such that 
$z= y \mod{M}.$

Let $\ncells(K \cdot \lattice^d(M) + \mathbf{a}, M, z_0)=0$, 
then $\ncells(K \cdot \lattice^d(M) + \mathbf{a}, N, y_0)=0$ for all $z_0= y_0 \mod{M}.$

Let $\ncells(K \cdot \lattice^d(M) + \mathbf{a}, M, z_1)\ne0$, then
there exists a non-empty subset of the lattice $K \cdot \lattice^d(M) + \mathbf{a}$ such that for each of its cells the sum of squared coordinates modulo $M$ is equal to $k_1$. That is, in this subset, the sums of squares of cell coordinates differ from each other by multiples of $M$. But we came to the conclusion that these quantities are multiples of $K$ as well. Therefore, they are multiples of the least common multiple of $M$ and $K$. Since $M$ and $K$ are coprime, these quantities are multiples of $N = MK$ (\ref{eq56}). Since $y_1\in\mathbb{Z}(N)$, among all $K$ distinct $y_1\equiv z_1 \mod{M}$ there is one value $\tilde y_1$ for which
\begin{equation}\label{eq72}
  \ncells(K \cdot \lattice^d(M) + \mathbf{a}, N, \tilde y_1)
 =\ncells(K \cdot \lattice^d(M) + \mathbf{a}, M, z_1),
\end{equation}
for the remaining $y_1\ne\tilde y_1$ we get
$\ncells(K \cdot \lattice^d(M) + \mathbf{a}, N, y_1)=0$.

\begin{figure}[H]
	\centering
	\includegraphics[width=1\linewidth]{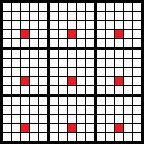}
	\caption{Решётка $\lattice^2(N)$. Здесь для примера $N=15$, $M=3$, ${K=5}$. Красным цветом показан пример множества $K \cdot \lattice^2(M) + \mathbf{a}$, где ${\mathbf{a} \in \lattice^2(K)}$.}
	\label{fig:p7}
\end{figure}

Therefore, due to (\ref{eq69})
\be\label{eq73}
  \zrf(K \cdot \lattice^d(M) + \mathbf{a}, N, M),
\ee
using the lemma \ref{l5} (on sum of lattices) we obtain
\be\label{eq73+}
  \zrf(\lattice^d(N), N, M).
\ee

By permutation of $M$ and $K$, we similarly prove that
\begin{equation}\label{eq74}
  \zrf(M \cdot \lattice^d(K) + \mathbf{a}, N, K),\quad \mathbf{a}\in \lattice^d(M),
\end{equation}
using the lemma \ref{l5} (on sum of lattices) we obtain
\be\label{eq74+}
   \zrf(lattice^d(N), N, K).
\ee

Since $M$ and $K$ are coprime and $N = MK$ (\ref{eq56}), their least common multiple is equal to $N$, so by (\ref{eq73+}) and (\ref{eq74+}) lattice $\lattice^d(N)$ is zerofied modulo $N$:
\begin{equation}\label{eq75}
  \zrf(\lattice^d(N),N,N).
\end{equation}


\end{document}